\documentclass[sigplan,screen,nonacm]{acmart}
\usepackage{paper}
\usepackage{fontmathlocal}

\makeatletter
\@ifundefined{lhs2tex.lhs2tex.sty.read}%
  {\@namedef{lhs2tex.lhs2tex.sty.read}{}%
   \newcommand\SkipToFmtEnd{}%
   \newcommand\EndFmtInput{}%
   \long\def\SkipToFmtEnd#1\EndFmtInput{}%
  }\SkipToFmtEnd

\newcommand\ReadOnlyOnce[1]{\@ifundefined{#1}{\@namedef{#1}{}}\SkipToFmtEnd}
\usepackage{amstext}
\usepackage{amssymb}
\usepackage{stmaryrd}
\DeclareFontFamily{OT1}{cmtex}{}
\DeclareFontShape{OT1}{cmtex}{m}{n}
  {<5><6><7><8>cmtex8
   <9>cmtex9
   <10><10.95><12><14.4><17.28><20.74><24.88>cmtex10}{}
\DeclareFontShape{OT1}{cmtex}{m}{it}
  {<-> ssub * cmtt/m/it}{}

\DeclareFontShape{OT1}{cmtt}{bx}{n}
  {<5><6><7><8>cmtt8
   <9>cmbtt9
   <10><10.95><12><14.4><17.28><20.74><24.88>cmbtt10}{}
\DeclareFontShape{OT1}{cmtex}{bx}{n}
  {<-> ssub * cmtt/bx/n}{}

\newcommand{\Conid}[1]{\mathit{#1}}
\newcommand{\Varid}[1]{\mathit{#1}}
\newcommand{\anonymous}{\kern0.06em \vbox{\hrule\@width.5em}}

\newcommand{\bind}{\mathbin{>\!\!\!>\mkern-6.7mu=}}
\newcommand{\sequ}{\mathbin{>\!\!\!>}}

\usepackage{polytable}

\@ifundefined{mathindent}%
  {\newdimen\mathindent\mathindent\leftmargini}%
  {}%

\def\resethooks{%
  \global\let\SaveRestoreHook\empty
  \global\let\ColumnHook\empty}
\newcommand*{\savecolumns}[1][default]%
  {\g@addto@macro\SaveRestoreHook{\savecolumns[#1]}}
\newcommand*{\restorecolumns}[1][default]%
  {\g@addto@macro\SaveRestoreHook{\restorecolumns[#1]}}
\newcommand*{\aligncolumn}[2]%
  {\g@addto@macro\ColumnHook{\column{#1}{#2}}}

\resethooks

\newcommand{\onelinecommentchars}{\quad-{}- }
\newcommand{\commentbeginchars}{\enskip\{-}
\newcommand{\commentendchars}{-\}\enskip}

\newcommand{\visiblecomments}{%
  \let\onelinecomment=\onelinecommentchars
  \let\commentbegin=\commentbeginchars
  \let\commentend=\commentendchars}

\newcommand{\invisiblecomments}{%
  \let\onelinecomment=\empty
  \let\commentbegin=\empty
  \let\commentend=\empty}

\visiblecomments

\newlength{\blanklineskip}
\newcommand{\hsindent}[1]{\quad}
\let\hspre\empty
\let\hspost\empty

\EndFmtInput
\makeatother
\ReadOnlyOnce{polycode.fmt}%
\makeatletter

\newcommand{\hsnewpar}[1]%
  {{\parskip=0pt\parindent=0pt\par\vskip #1\noindent}}

\newcommand{\hscodestyle}{}

\newcommand{\sethscode}[1]%
  {\expandafter\let\expandafter\hscode\csname #1\endcsname
   \expandafter\let\expandafter\endhscode\csname end#1\endcsname}

  {\par\noindent
   \advance\leftskip\mathindent
   \hscodestyle
   \let\\=\@normalcr
   \let\hspre\(\let\hspost\)%
   \pboxed}%
  {\endpboxed\)%
   \par\noindent
   \ignorespacesafterend}

  {\hsnewpar\abovedisplayskip
   \advance\leftskip\mathindent
   \hscodestyle
   \let\hspre\(\let\hspost\)%
   \pboxed}%
  {\endpboxed%
   \hsnewpar\belowdisplayskip
   \ignorespacesafterend}

  {\hsnewpar\abovedisplayskip
   \advance\leftskip\mathindent
   \hscodestyle
   \let\\=\@normalcr
   \(\pboxed}%
  {\endpboxed\)%
   \hsnewpar\belowdisplayskip
   \ignorespacesafterend}

\newcommand{\plainhs}{\sethscode{plainhscode}}

\plainhs

  {\hsnewpar\abovedisplayskip
   \advance\leftskip\mathindent
   \hscodestyle
   \let\\=\@normalcr
   \(\parray}%
  {\endparray\)%
   \hsnewpar\belowdisplayskip
   \ignorespacesafterend}

  {\parray}{\endparray}

  {\(\parray}{\endparray\)}

\def\codeframewidth{\arrayrulewidth}
\RequirePackage{calc}

  {\parskip=\abovedisplayskip\par\noindent
   \hscodestyle
   \arrayrulewidth=\codeframewidth
   \tabular{@{}|p{\linewidth-2\arraycolsep-2\arrayrulewidth-2pt}|@{}}%
   \hline\framedhslinecorrect\\{-1.5ex}%
   \let\endoflinesave=\\
   \let\\=\@normalcr
   \(\pboxed}%
  {\endpboxed\)%
   \framedhslinecorrect\endoflinesave{.5ex}\hline
   \endtabular
   \parskip=\belowdisplayskip\par\noindent
   \ignorespacesafterend}

\newcommand{\framedhslinecorrect}[2]%
  {#1[#2]}

  {\(\def\column##1##2{}%
   \let\>\undefined\let\<\undefined\let\\\undefined
   \newcommand\>[1][]{}\newcommand\<[1][]{}\newcommand\\[1][]{}%
   \def\fromto##1##2##3{##3}%
   }{\) }%

  {\let\orighscode=\hscode
   \let\origendhscode=\endhscode
   \def\endhscode{\def\hscode{\endgroup\def\@currenvir{hscode}\\}\begingroup}
   \orighscode\def\hscode{\endgroup\def\@currenvir{hscode}}}%
  {\origendhscode
   \global\let\hscode=\orighscode
   \global\let\endhscode=\origendhscode}%

\makeatother
\EndFmtInput
\AtBeginDocument{%
  \providecommand\BibTeX{{%
    \normalfont B\kern-0.5em{\scshape i\kern-0.25em b}\kern-0.8em\TeX}}}

\acmConference{}{}{}
\acmBooktitle{}
\acmPrice{}
\acmISBN{}

\begin{document}

\pagestyle{plain}

\title{Equational Reasoning for MTL Type Classes}

\author{H\"armel Nestra}
\email{harmel.nestra@ut.ee}
\orcid{https://orcid.org/0000-0001-7050-7171}
\affiliation{%
  \institution{University of Tartu}
  \institution{Institute of Computer Science}
  \streetaddress{Narva Rd. 18}
  \city{Tartu}
  \postcode{51009}
  \country{Estonia}
}

\begin{abstract}
Ability to use definitions occurring in the code directly 
in equational reasoning is one of the 
key strengths of functional programming. This is impossible in the case of
Haskell type class methods unless a particular instance type is specified,
since class methods can be defined differently for different instances. 
To allow uniform reasoning for all instances, many type classes in the
Haskell library come along with laws (axioms), specified in comments, 
that all instances 
are expected to follow (albeit Haskell is unable to force it). For the type
classes introduced in the Monad Transformer Library (MTL), such
laws have not been specified; nevertheless, some sets of axioms 
have occurred in the literature and the Haskell mailing lists. 
This paper investigates sets of laws usable for equational reasoning about 
methods of the type classes \ensuremath{\Conid{MonadReader}} and \ensuremath{\Conid{MonadWriter}} and also reviews 
analogous earlier proposals for the classes \ensuremath{\Conid{MonadError}} and \ensuremath{\Conid{MonadState}}. 
For both \ensuremath{\Conid{MonadReader}} and \ensuremath{\Conid{MonadWriter}}, an equivalence
result of two alternative axiomatizations in terms of different sets of
operations is established. As a sideline, patterns in the choice of methods of
different classes are noticed which may inspire new developments in MTL.
\end{abstract}



\maketitle

\section{Introduction}\label{intro}

Equational reasoning is often presented as one of the key benefits of 
functional programming. Definitions in the source code 
provide us basic equalities to rely on, and referential 
transparency in pure functional languages like Haskell allows one to 
safely replace terms by equal terms within any context.
Definitional equality works like in mathematics. 

As Haskell type class methods are
defined newly for every instance type, equational reasoning relying on
method definitions in the code is type specific. In order to create 
uniform proofs for all instances of a class, one must use assumptions in the
form of equations, called \emph{axioms} or \emph{laws}, which are 
not grounded on the source code. Many type classes defined in the Haskell
library come along with such axioms which every instance of the class 
is supposed (though not forced) to satisfy. For example, the method
\ensuremath{\Varid{fmap}} of all instances of the 
class \ensuremath{\Conid{Functor}} is supposed to preserve composition and identities, the
methods of the class \ensuremath{\Conid{Monad}} should satisfy the standard monad laws of
category theory, etc.

The class \ensuremath{\Conid{Monad}} offers a uniform interface for effectful computations of
various kinds. The Monad Transformer Library (MTL) which is a part
of GHC standard libraries provides a lot of subclasses of \ensuremath{\Conid{Monad}} 
with specific operation interfaces for different effects. It implements 
a classic approach dating back to 
Liang et al.~\cite{DBLP:conf/popl/LiangHJ95}, 
Jones~\cite{DBLP:conf/afp/Jones95} and Hutton and Meijer~\cite{HuttonM1996}, 
yet its type classes have 
no universally accepted laws for equational reasoning. Gibbons and
Hinze~\cite{DBLP:conf/icfp/GibbonsH11} advocate equational reasoning for
effectful computations, in particular in the case of some MTL classes, but
without analyzing the choice of the laws or studying alternatives.

Recently, laws for MTL type classes have gained some attention in 
research \cite{DBLP:journals/corr/abs-1810-13430,DBLP:conf/ictac/Nestra19,DBLP:conf/mpc/AffeldtNS19} 
and in Haskell Libraries mailing list \cite{Libraries} which 
tells about rising interest in this topic. This paper studies 
several equational axiomatizations of monadic computations with various 
effects---exceptions, environment, logging (writer) and mutable
state. For environment and logging, we propose a few
alternative sets of laws, prove their equivalence, and prove correctness of 
these laws for monads built up using 
the exception, reader, writer and state monad 
transformers. For exceptions and state, we review similar results of
previous work. The axiomatics considered here address one effect
each; axiomatizing of combinations of different effects might be 
an interesting topic of future work.

We avoid a premature conclusion to have found 
\qquot{the right axiomatizations}
for the four type classes. Firstly, validity on four types 
of effects might not provide sufficient evidence for declaring the laws 
universal enough. Secondly, some of our axioms assume operations with
types that are impossible in MTL. Nevertheless, 
such sets of axioms can be useful for proving equivalences of legal
programs. 
Providing a means for writing such proofs is the primary goal of developing
the axiomatizations. Solving problems by going beyond the bounds imposed by 
the problem setting is neither unsound nor original. Quite 
analogously to our case, Hutton and
Fulger~\cite{Hutton08reasoningabout} lift type restrictions imposed by the
context in their proofs of equivalences of effectful programs.

Like in \cite{DBLP:conf/icfp/GibbonsH11}, we ignore partiality that 
may break the laws (see Jeuring et al.~\cite{DBLP:conf/haskell/JeuringJA12}
for proof that the state monad transformer does not preserve monad laws for
bottom cases, and Huffman~\cite{DBLP:conf/icfp/Huffman12} for similar results
for the exception and writer monad transformer). 
Danielsson et~al.~\cite{DBLP:conf/popl/DanielssonHJG06} show that ignoring
partiality in equational reasoning is justified. We use 
notation from category theory throughout the paper mainly for achieving more
concise formulae rather than for generality. Most results are established for
category $\Set$ only.

In Sect.~\ref{prelim}, we help the reader with preliminaries
from category theory; Sect.~\ref{mtl} gives a short introduction to MTL. 
In Sect.~\ref{except}, we review previous work on laws of
exceptions. In Sect.~\ref{reader}, we study axiomatics for equational
reasoning about reader monad operations. Section~\ref{writer} is devoted to
writer monads. We develop an abstract treatment of monads equipped with some
additional operations in general category theoretical setting, which luckily
applies to the writer case. The set of operations most natural for
investigating at the abstract level does not coincide with the actual 
method set specified in MTL but, as the latter is directly
expressible in terms of the former, the results are applicable to Haskell. 
Section~\ref{state} contains a brief review of 
previously proposed axiomatics of stateful computations. Section~\ref{conc}
concludes. 

\section{Preliminaries from Category Theory}\label{prelim}

In mathematics, a \emph{category} consists of a
collection of \emph{objects}, along with a collection of \emph{morphisms}
working between every ordered pair of objects, satisfying the following 
properties:
\begin{itemize}
\item For every object $\obx$, there exists a distinguished morphism
$\id_{\obx}$ from $\obx$ to $\obx$ called the \emph{identity} of $\obx$;
\item For every triple $(\obx,\oby,\obz)$ of objects and morphisms $f$ from
$\obx$ to $\oby$ and $g$ from $\oby$ to $\obz$, there exists a morphism
$g\icp f$ from $\obx$ to $\obz$ called the \emph{composition} of $f$ and $g$;
\item The composition as an operation is associative, and each identity 
morphism $\id_{\obx}$ works as both a left and right unit of composition.
\end{itemize}
The claim that $f$ is a morphism from object $\obx$ to object $\oby$ is
denoted by $f\oftyp\obx\to\oby$, where $\obx\to\oby$ is sometimes called
the \emph{type} of $f$. The object in the subscript of $\id$ can be omitted 
if determined by the context. Assume 
the binding precedence of~$\icp$ being
higher than that of any other binary operator.

It is common to use notions of category theory when talking about Haskell. 
The data types are playing the role of objects and definable functions from
one type to another are the morphisms between these types as objects. 
Identity functions and function composition play the corresponding role. 
Laziness of Haskell, along with the presence of \ensuremath{\Varid{seq}}, does 
not allow all properties of category to be fully satisfied, nor do the standard
categorical constructions in Haskell (some of which are introduced below) 
fully meet their definition given in category theory; however, as shown by 
Danielsson et al.~\cite{DBLP:conf/popl/DanielssonHJG06}, it is justified to 
ignore the deviations for practical purposes. 
We will use both the notation of Haskell and that of category theory
throughout the paper, 
depending on which is more concise and readable at a particular place. 

Pair types \ensuremath{(\Varid{x},\Varid{y})} of Haskell correspond to the binary products
in category theory where they are denoted by cross. In category
theory, a product $\obx\times\oby$ of objects $\obx$ and $\oby$ comes along 
with morphisms
$\exl\oftyp\obx\times\oby\to\obx$, $\exr\oftyp\obx\times\oby\to\oby$ and 
an operator $\both$ mapping every pair of morphisms $f\oftyp\obz\to\obx$
and $g\oftyp\obz\to\oby$ to a morphism 
$f\both g\oftyp\obz\to\obx\times\oby$. Thereby, they must meet the laws
$\exl\icp(f\both g)=f$, $\exr\icp(f\both g)=g$ and 
$\exl\icp h\both\exr\icp h=h$. In Haskell, the projection functions 
\ensuremath{\Varid{fst}} and \ensuremath{\Varid{snd}} stand 
for $\exl$ and $\exr$; for $f\both g$ one can take the function 
that applies both $f$ and $g$ to its argument and returns the results 
as a pair. 

Similarly, the types \ensuremath{\Conid{Either}\;\Varid{x}\;\Varid{y}} of Haskell correspond to binary sums of 
category theory where they are written by plus. A sum $\obx+\oby$ of objects
comes along with morphisms $\inl\oftyp\obx\to\obx+\oby$,
$\inr\oftyp\oby\to\obx+\oby$ and an operator $\either$ that maps every pair
of morphisms $f\oftyp\obx\to\obz$ and $g\oftyp\oby\to\obz$ to a morphism
$f\either g\oftyp\obx+\oby\to\obz$. Thereby, they must satisfy the equations
$(f\either g)\icp\inl=f$, $(f\either g)\icp\inr=g$ and
$h\icp\inl\either h\icp\inr=h$. In Haskell, $\inl$ and $\inr$ are written as
\ensuremath{\Conid{Left}} and \ensuremath{\Conid{Right}}, whereas the intended behavior of $\either$ is 
captured by the library function \ensuremath{\Varid{either}}.

The \ensuremath{\Conid{Functor}} type class in Haskell is introduced by
\begin{hscode}\SaveRestoreHook
\column{B}{@{}>{\hspre}l<{\hspost}@{}}%
\column{5}{@{}>{\hspre}l<{\hspost}@{}}%
\column{E}{@{}>{\hspre}l<{\hspost}@{}}%
\>[B]{}\mathbf{class}\;\Conid{Functor}\;\Varid{f}\;\mathbf{where}{}\<[E]%
\\
\>[B]{}\hsindent{5}{}\<[5]%
\>[5]{}\Varid{fmap}\mathbin{::}(\Varid{a}\to \Varid{b})\to \Varid{f}\;\Varid{a}\to \Varid{f}\;\Varid{b}{}\<[E]%
\\
\>[B]{}\hsindent{5}{}\<[5]%
\>[5]{}\mbox{\commentbegin  ... other stuff omitted ...  \commentend}{}\<[E]%
\ColumnHook
\end{hscode}\resethooks
In category theory, a \emph{functor}~$\frf$ is a structure-preserving
mapping between categories, i.e., a mapping of objects of one category 
to objects of the other category and morphisms
of type $\obx\to\oby$ for any objects $\obx,\oby$ of the first category 
to morphisms of type $\frf\ap\obx\to\frf\ap\oby$ in the second category, 
satisfying the laws $\frf\ap\id=\id$ and 
$\frf\ap(g\icp f)=\frf\ap g\icp\frf\ap f$. The Haskell class method \ensuremath{\Varid{fmap}} 
corresponds to the mapping of morphisms while the mapping of objects is
implemented by the type constructor~\ensuremath{\Varid{f}}. The two functor laws are 
expected to be fulfilled by every instance of the \ensuremath{\Conid{Functor}} type
class.

The library of Haskell defines also the class \ensuremath{\Conid{Bifunctor}} analogous to
\ensuremath{\Conid{Functor}} but applying to binary type constructors:
\begin{hscode}\SaveRestoreHook
\column{B}{@{}>{\hspre}l<{\hspost}@{}}%
\column{5}{@{}>{\hspre}l<{\hspost}@{}}%
\column{E}{@{}>{\hspre}l<{\hspost}@{}}%
\>[B]{}\mathbf{class}\;\Conid{Bifunctor}\;\Varid{p}\;\mathbf{where}{}\<[E]%
\\
\>[B]{}\hsindent{5}{}\<[5]%
\>[5]{}\Varid{bimap}\mathbin{::}(\Varid{a}\to \Varid{b})\to (\Varid{c}\to \Varid{d})\to \Varid{p}\;\Varid{a}\;\Varid{c}\to \Varid{p}\;\Varid{b}\;\Varid{d}{}\<[E]%
\\
\>[B]{}\hsindent{5}{}\<[5]%
\>[5]{}\mbox{\commentbegin  ... other stuff omitted ...  \commentend}{}\<[E]%
\ColumnHook
\end{hscode}\resethooks
The functor laws are assumed to hold for both argument types. In
category theory, bifunctors can be seen as functors whose domain is the direct 
product of two categories. Both the binary product and binary
sum, considered above as operations on objects, can be extended to morphisms 
by defining $f\times g=f\icp\exl\both g\icp\exr$ and 
$f+g=\inl\icp f\either\inr\icp g$, as the result of which one obtains  
bifunctors.

A \emph{transformation} between functors $\frf$ and $\frg$ 
is a family~$\tau$ of morphisms 
$\tau_{\obx}\oftyp\frf\ap\obx\to\frg\ap\obx$ 
for every object $\obx$. In the programming point of view, transformations
are polymorphic functions. 
A transformation is called \emph{natural} if it preserves
structure embedded in the functors, i.e., satisfies for every morphism~$f$ 
the equation
$\tau\icp\frf\ap f=\frg\ap f\icp\tau$. Like in the case of identity
morphisms (which altogether constitute, of course, a natural transformation 
between $\fri$~and~$\fri$ where $\fri$ is the identity functor leaving
everything in place), 
the subscript of~$\tau$ is omitted when it is clear from context.

The \ensuremath{\Conid{Monad}} type class which is a subclass of \ensuremath{\Conid{Functor}} declares methods
\begin{hscode}\SaveRestoreHook
\column{B}{@{}>{\hspre}l<{\hspost}@{}}%
\column{5}{@{}>{\hspre}l<{\hspost}@{}}%
\column{13}{@{}>{\hspre}c<{\hspost}@{}}%
\column{13E}{@{}l@{}}%
\column{17}{@{}>{\hspre}l<{\hspost}@{}}%
\column{30}{@{}>{\hspre}l<{\hspost}@{}}%
\column{36}{@{}>{\hspre}l<{\hspost}@{}}%
\column{E}{@{}>{\hspre}l<{\hspost}@{}}%
\>[5]{}(\bind ){}\<[13]%
\>[13]{}\mathbin{::}{}\<[13E]%
\>[17]{}\Varid{m}\;\Varid{a}\to (\Varid{a}\to \Varid{m}\;\Varid{b}){}\<[36]%
\>[36]{}\to \Varid{m}\;\Varid{b}{}\<[E]%
\\
\>[5]{}(\sequ ){}\<[13]%
\>[13]{}\mathbin{::}{}\<[13E]%
\>[17]{}\Varid{m}\;\Varid{a}\to {}\<[30]%
\>[30]{}\Varid{m}\;\Varid{b}{}\<[36]%
\>[36]{}\to \Varid{m}\;\Varid{b}{}\<[E]%
\\
\>[5]{}\Varid{return}{}\<[13]%
\>[13]{}\mathbin{::}{}\<[13E]%
\>[17]{}\Varid{a}\to \Varid{m}\;\Varid{a}{}\<[E]%
\ColumnHook
\end{hscode}\resethooks
along with the default definition \ensuremath{\Varid{x}\sequ\Varid{k}\mathrel{=}\Varid{x}\bind \lambda \anonymous \to \Varid{k}};
the type variable \ensuremath{\Varid{m}} stands for arbitrary instance type
of \ensuremath{\Conid{Monad}}. 
The operator \ensuremath{\bind } is pronounced \emph{bind}, whereas
\ensuremath{\Varid{return}} is called \emph{unit} of the monad. In category
theory, bind is usually denoted by $\hole^*$ and the order of its arguments is
reversed; so if $\frm$~is a monad then $\hole^*$ maps morphisms of type
$\obx\to\frm\ap\oby$ to morphisms of type $\frm\ap\obx\to\frm\ap\oby$. 
The Haskell operator \ensuremath{\sequ } 
is a special case of bind with constant function as argument.

In category theory, monad operations must satisfy the unit laws
$k^*\icp\invok=k$ and $\invok^*=\id$ along with 
associativity $l^*\icp k^*=(l^*\icp k)^*$, and 
the functor must be expressible via monad operations by 
$\frm\ap f=(\invok\icp f)^*$. The same axioms are expected to be met by all
instances of the \ensuremath{\Conid{Monad}} type class. Category theorists often define monads 
via $\join\oftyp\frm\ap\frm\ap\obx\to\frm\ap\obx$ instead of bind. The join and
bind operations are expressed in terms of each other by $\join=\id^*$ and 
$k^*=\join\icp\frm\ap k$. The axiom set for this approach 
equivalent to the one given above consists of:
\begin{itemize}
\item The two functor laws for $\frm$;
\item The naturality laws $\invok\icp f=\frm\ap f\icp\invok$ and
$\join\icp\frm\ap\frm\ap f=\frm\ap f\icp\join$;
\item The \emph{coherence} axioms 
$\join\icp\frm\ap\join=\join\icp\join$,\break
$\join\icp\frm\ap\invok=\id$ and $\join\icp\invok=\id$.
\end{itemize}
The morphisms $\invok$ and $\join$ are usually denoted by $\eta$ and $\mu$,
respectively. 
\emph{Identity monad} $\fri$ is the simplest monad where the mapping of
objects, mappings of morphisms, monad unit, bind and join are all identities.

\emph{Relative monads on a base functor $\frj$} were introduced by
Altenkirch 
et~al.~\cite{DBLP:conf/fossacs/AltenkirchCU10,DBLP:journals/corr/AltenkirchCU14}. 
Relative monads are pairs of functors $(\frj,\frm)$ equipped with 
bind and unit operations whose types are more general 
than those of the corresponding monad operations: namely, the unit must have 
type $\frj\ap\obx\to\frm\ap\obx$, and bind takes morphisms of type
$\frj\ap\obx\to\frm\ap\oby$ to morphisms of type
$\frm\ap\obx\to\frm\ap\oby$. The functor $\frm$ must be expressible via
these operations by $\frm\ap f=(\invok\icp\frj\ap f)^*$. Other relative
monad laws look the same as the monad laws of $\invok$ and $\hole^*$. 
Monads are relative monads where $\frj=\fri$. Due to
the different types, the alternative representation via $\join$ is
impossible for relative monads in general. 

A \emph{monad morphism} from $\frm$ to $\frm'$ is a structure-preserving
transformation between these monads, i.e., 
$\sigma\oftyp\frm\ap\obx\to\frm'\ap\obx$ such that $\sigma\icp\invok=\invok$
and $\sigma\icp k^*=(\sigma\icp k)^*\icp\sigma$, where $\invok$ and $\hole^*$
in the left-hand sides belong to $\frm$ and those in the right-hand sides
belong to $\frm'$. Analogously, one can define relative monad morphisms.

The term \emph{pointed functor} is sometimes used for denoting
functors $\frf$ equipped with a unit $\invok\oftyp\obx\to\frf\ap\obx$ 
(for any~$\obx$) but without monad bind. The unit is assumed to meet 
the naturality law $\invok\icp f=\frf\ap f\icp\invok$. The term
dates back to at least 
Lenisa et~al.~\cite{DBLP:journals/entcs/LenisaPW00}. 
We prefer to call $\invok$ of a pointed functor its \emph{point} 
rather than unit, since calling something a unit normally assumes a certain 
relationship with another (binary) operation 
(like the unit laws relating unit and bind in the case of monads) 
which is missing in the general case of pointed functors. 
We will occasionally call the unit of a relative monad also \emph{point} 
or \emph{relative point}.

\section{A Brief Introduction to MTL Classes}\label{mtl}

Monads as a framework suitable for embedding computational effects were
discovered and advocated by Moggi~\cite{Moggi1990}. The same work studied 
constructs in category theory that we now know as monad
transformers. Using monad transformers in Haskell was proposed by 
Liang et al.~\cite{DBLP:conf/popl/LiangHJ95}, 
Jones~\cite{DBLP:conf/afp/Jones95} and Hutton and Meijer~\cite{HuttonM1996}. 
The Haskell MTL has been built upon ideas of those papers. It 
defines the exception, reader, writer, and state monad transformers as follows:
\begin{hscode}\SaveRestoreHook
\column{B}{@{}>{\hspre}l<{\hspost}@{}}%
\column{18}{@{}>{\hspre}l<{\hspost}@{}}%
\column{21}{@{}>{\hspre}l<{\hspost}@{}}%
\column{24}{@{}>{\hspre}l<{\hspost}@{}}%
\column{27}{@{}>{\hspre}l<{\hspost}@{}}%
\column{38}{@{}>{\hspre}l<{\hspost}@{}}%
\column{E}{@{}>{\hspre}l<{\hspost}@{}}%
\>[B]{}\mathbf{newtype}\;\Conid{ExceptT}\;{}\<[18]%
\>[18]{}\Varid{e}\;{}\<[21]%
\>[21]{}\Varid{m}\;{}\<[24]%
\>[24]{}\Varid{a}{}\<[27]%
\>[27]{}\mathrel{=}\Conid{ExceptT}\;{}\<[38]%
\>[38]{}(\Varid{m}\;(\Conid{Either}\;\Varid{e}\;\Varid{a})){}\<[E]%
\\
\>[B]{}\mathbf{newtype}\;\Conid{ReaderT}\;{}\<[18]%
\>[18]{}\Varid{r}\;{}\<[21]%
\>[21]{}\Varid{m}\;{}\<[24]%
\>[24]{}\Varid{a}{}\<[27]%
\>[27]{}\mathrel{=}\Conid{ReaderT}\;{}\<[38]%
\>[38]{}(\Varid{r}\to \Varid{m}\;\Varid{a}){}\<[E]%
\\
\>[B]{}\mathbf{newtype}\;\Conid{WriterT}\;{}\<[18]%
\>[18]{}\Varid{w}\;{}\<[21]%
\>[21]{}\Varid{m}\;{}\<[24]%
\>[24]{}\Varid{a}{}\<[27]%
\>[27]{}\mathrel{=}\Conid{WriterT}\;{}\<[38]%
\>[38]{}(\Varid{m}\;(\Varid{a},\Varid{w})){}\<[E]%
\\
\>[B]{}\mathbf{newtype}\;\Conid{StateT}\;{}\<[18]%
\>[18]{}\Varid{s}\;{}\<[21]%
\>[21]{}\Varid{m}\;{}\<[24]%
\>[24]{}\Varid{a}{}\<[27]%
\>[27]{}\mathrel{=}\Conid{StateT}\;{}\<[38]%
\>[38]{}(\Varid{s}\to \Varid{m}\;(\Varid{a},\Varid{s})){}\<[E]%
\ColumnHook
\end{hscode}\resethooks
(There are more transformers but we only study these four.) 
Every transformer assumes a type constructor as its second parameter~\ensuremath{\Varid{m}};
provided that it is a monad, application of the transformer produces a new
monad. Substituting the identity monad for \ensuremath{\Varid{m}}, we obtain
\emph{the error, reader, writer, and state monads}, respectively. 
Using the language of category theory, we can define these monads as
$\Except\ap(\obe,\obx)=\obe+\obx$, $\Reader\ap(\obr,\obx)=\obr\to\obx$, 
$\Writer\ap(\obw,\obx)=\obx\times\obw$, and
$\State\ap(\obs,\obx)=\obs\to\obx\times\obs$. (Still we use the Haskell 
notation $\oba\to\obb$ for function space which deviates from the notations
of exponential objects used in category theory.)

Every transformer adds new effects to the monad while keeping the
previously existing effects in force and available:
\begin{itemize}
\item The exception 
monad transformer adds capability of dealing with errors, where 
errors are values of the type given as the first parameter of the transformer;
\item The state monad transformer enables one to use a hidden 
mutable state for computation, 
where states have type given as the first parameter of the transformer;
\item The reader monad transformer makes it possible to use hidden 
\qquot{environments} of the type given as the first parameter of the
transformer;
\item The writer monad transformer introduces the ability of 
information logging throughout the computation, where the data items 
written into the log
are of the type determined by the first parameter of the transformer.
\end{itemize}
We omit the details of defining the relevant operations and their
propagation
through chains of monad transformer applications; they are not needed for
understanding the paper.

It is convenient to have every effect introduced by 
some monad transformer accessible via a fixed interface irrespectively of
other monad transformers applied. For that reason, MTL defines 
type classes \ensuremath{\Conid{MonadError}}, \ensuremath{\Conid{MonadReader}}, \ensuremath{\Conid{MonadWriter}} and
\ensuremath{\Conid{MonadState}} (and others for effects not considered here). For
instance, \ensuremath{\Conid{MonadError}} defines the interface of exception throwing 
and handling, \ensuremath{\Conid{MonadState}} specifies the interface for stateful 
computation, etc. Each of the following sections
\ref{except}--\ref{state} discusses one of these type classes, 
providing also the definition details.

\section{Laws of Exception Handling}\label{except}

MTL defines the \ensuremath{\Conid{MonadError}} class along with methods 
for throwing and catching of exceptions as follows:
\begin{hscode}\SaveRestoreHook
\column{B}{@{}>{\hspre}l<{\hspost}@{}}%
\column{5}{@{}>{\hspre}l<{\hspost}@{}}%
\column{17}{@{}>{\hspre}l<{\hspost}@{}}%
\column{E}{@{}>{\hspre}l<{\hspost}@{}}%
\>[B]{}\mathbf{class}\;(\Conid{Monad}\;\Varid{m})\Rightarrow \Conid{MonadError}\;\Varid{e}\;\Varid{m}\mid \Varid{m}\to \Varid{e}\;\mathbf{where}{}\<[E]%
\\
\>[B]{}\hsindent{5}{}\<[5]%
\>[5]{}\Varid{throwError}{}\<[17]%
\>[17]{}\mathbin{::}\Varid{e}\to \Varid{m}\;\Varid{a}{}\<[E]%
\\
\>[B]{}\hsindent{5}{}\<[5]%
\>[5]{}\Varid{catchError}{}\<[17]%
\>[17]{}\mathbin{::}\Varid{m}\;\Varid{a}\to (\Varid{e}\to \Varid{m}\;\Varid{a})\to \Varid{m}\;\Varid{a}{}\<[E]%
\ColumnHook
\end{hscode}\resethooks
Here \ensuremath{\Varid{e}} denotes the type of exceptions; it is a parameter of the class but not
of \ensuremath{\Varid{m}}. Gibbons and Hinze~\cite{DBLP:conf/icfp/GibbonsH11} consider axioms of
exception handling in a narrower setting that does not involve 
an exception type (it is equivalent to \ensuremath{\Conid{MonadError}}
with \ensuremath{\Varid{e}\mathrel{=}()}). Their axioms state that catch is associative, whereas failure
(throwing the exception) is its unit and a left zero of~\ensuremath{\sequ } 
as well. Both Malakhovski~\cite{DBLP:journals/corr/abs-1810-13430} and 
the author~\cite{DBLP:conf/ictac/Nestra19} assume a wider setting with 
the type \ensuremath{\Varid{e}} being an additional parameter of the
functor \ensuremath{\Varid{m}}, 
replacing the latter with a bifunctor~$\frf$. The error throwing 
function has type $\obe\to\frf\ap(\obe,\oba)$ and return obtains type
$\oba\to\frf\ap(\obe,\oba)$. 
The catch function in~\cite{DBLP:conf/ictac/Nestra19} has type 
$(\obe\to\frf\ap(\obe',\oba))\to\frf\ap(\obe,\oba)\to\frf\ap(\obe',\oba)$
similarly to bind which has type
$(\oba\to\frf\ap(\obe,\oba'))\to\frf\ap(\obe,\oba)\to\frf\ap(\obe,\oba')$;
in~\cite{DBLP:journals/corr/abs-1810-13430}, the order of arguments of these
functions is reversed.

Following~\cite{DBLP:conf/ictac/Nestra19}, we denote the bind and catch
operations by $\hole\bnd$ and $\hole\cch$, respectively, and denote the error
throwing function by $\rais$. 
That work proposes the following
axioms for bind, among which the first three are standard monad axioms, the
fourth one is a generalization of the usual monad zero law by introducing an
exception parameter, and the last one
establishes that any mapping of exceptions by the bifunctor is a bind
homomorphism:
\lawarray{
&k\bnd\icp\invok
&&\al{=}k
\\
&(\invok
\icp f)\bnd
&&\al{=}\frf\ap(\id,f)
\\
&l\bnd\icp k\bnd
&&\al{=}(l\bnd\icp k)\bnd
\\
&k\bnd\icp\rais
&&\al{=}\rais
\\
&\frf\ap(h,\id)\icp k\bnd
&&\al{=}(\frf\ap(h,\id)\icp k)\bnd\icp\frf\ap(h,\id)
}
For catch, the dual laws are proposed:
\lawarray{
&k\cch\icp\rais
&&\al{=}k
\\
&(\rais
\icp h)\cch
&&\al{=}\frf\ap(h,\id)
\\
&l\cch\icp k\cch
&&\al{=}(l\cch\icp k)\cch
\\
&k\cch\icp\invok
&&\al{=}\invok
\\
&\frf\ap(\id,f)\icp k\cch
&&\al{=}(\frf\ap(\id,f)\icp k)\cch\icp\frf\ap(\id,f)
}
All but the last axiom of both blocks occur also 
in~\cite{DBLP:journals/corr/abs-1810-13430}. 
All axioms introduced so far are meaningful also in the standard MTL setting. 
In addition, \cite{DBLP:conf/ictac/Nestra19} considers the
following law for interchanging bind and catch, where
$\rho=\rais\either\invok$:
\lawarray{
&k\bnd\icp\rho\cch
&&\al{=}(\rais
\either k)\cch\icp(\frf\ap(\inl,\id)\icp k)\bnd
}
This law inherently exploits the two-parameter setting since $\rho\cch$
changes the exception type of the computation.

The following is established by \cite{DBLP:conf/ictac/Nestra19}, where
propagation of throw and catch through applications of 
transformers are assumed to be defined similarly to MTL:
\begin{itemize}
\item Any monad obtained by applying the error monad transformer to another 
monad (and considered as a bifunctor) satisfies all the given laws; 
\item Applying the error, reader, writer, or state monad transformer
preserves all the given laws.
\end{itemize}

The paper~\cite{DBLP:conf/ictac/Nestra19} defines the \emph{joint handle}
$\hole\hdl$ via $\hole\bnd$ and $\hole\cch$ by 
$k\hdl=k\bnd\icp\rho\cch\icp\frf\ap(\inr\icp\inl,\inr)$ and finds an
axiomatics for it that is equivalent to the set of 11 axioms described above.
Namely, it observes that $\rho\oftyp\obe+\oba\to\frf\ap(\obe,\oba)$ and 
$\hole\hdl\oftyp(\obe+\oba\to\frf\ap(\obe',\oba'))\to\frf\ap(\obe,\oba)\to\frf\ap(\obe',\oba')$, 
the types coinciding with those of relative monad unit and bind 
where the sum bifunctor is in the role of the base
functor~$\frj$. The axiomatics contains the relative monad laws
$k\hdl\icp\rho=k$, $\rho\hdl=\id$, and $(\rho\icp(h+f))\hdl=\frf(h,f)$;
the remaining law, associativity $l\hdl\icp k\hdl=(l\hdl\icp k)\hdl$, is not
valid in general, wherefore it is replaced with three special cases which
establish the desired equivalence. For
details, please see \cite{DBLP:conf/ictac/Nestra19}. 

We can classify all operations considered
in~\cite{DBLP:conf/ictac/Nestra19}, including those not discussed here, 
into three levels:
\begin{enumerate}
\item \emph{Point} operations that create structured objects from pure
values: $\rais$, $\invok$, and the joint unit $\rho$. 
(In~\cite{DBLP:conf/ictac/Nestra19}, the joint unit is denoted by~$\eta$. 
We preferred $\rho$ for the sake of uniform notation throughout this paper.)
\item \emph{Mixmap} 
$\phi\oftyp(\obe+\oba\to\obe'+\oba')\to\frf\ap(\obe,\oba)\to\frf\ap(\obe',\oba')$
and its special cases e.g.
$\fuser\oftyp\frf\ap(\obe+\oba,\oba)\to\frf\ap(\obe,\oba)$ and 
$\fusel\oftyp\frf\ap(\obe,\obe+\oba)\to\frf\ap(\obe,\oba)$ given by
equations
$\fusel=\phi(\inl\either\id)$ and $\fuser=\phi(\id\either\inr)$. These 
functions change, without side effects, 
the output value returned or thrown by their argument computation, but are
more general than standard \ensuremath{\Varid{fmap}} and \ensuremath{\Varid{bimap}} as they are able to 
\qquot{mix} the bifunctor arguments $\obe$ and $\oba$. (For example, $\fusel$
moves some errors from the right-hand argument to the left.) MTL defines no 
functions of the mixmap category with \ensuremath{\Conid{MonadError}} class
constraint 
but in the framework of~\cite{DBLP:conf/ictac/Nestra19} they can be expressed 
in terms of existing functions; for instance, 
$\fuser=\rho\cch$ and $\fusel=\rho\bnd$. 
\item \emph{Handle} functions. This level contains functions that
execute effectful computations in sequence, i.e., bind, catch and
the joint handle. This level subsumes the previous level as 
one can define mixmap in its general form in terms of the joint handle by 
$\phi(g)=(\rho\icp g)\hdl$. 
\end{enumerate}

\section{Reader Monads Equationally}\label{reader}

The class \ensuremath{\Conid{MonadReader}} is designed for encoding computations in an implicit 
environment. Computations in an environment can be seen as 
stateful computations with the state being immutable. 
The class is defined in MTL as follows, 
where \ensuremath{\Varid{r}} stands for the type of the environment:
\begin{hscode}\SaveRestoreHook
\column{B}{@{}>{\hspre}l<{\hspost}@{}}%
\column{5}{@{}>{\hspre}l<{\hspost}@{}}%
\column{15}{@{}>{\hspre}c<{\hspost}@{}}%
\column{15E}{@{}l@{}}%
\column{19}{@{}>{\hspre}l<{\hspost}@{}}%
\column{23}{@{}>{\hspre}l<{\hspost}@{}}%
\column{E}{@{}>{\hspre}l<{\hspost}@{}}%
\>[B]{}\mathbf{class}\;\Conid{Monad}\;\Varid{m}\Rightarrow \Conid{MonadReader}\;\Varid{r}\;\Varid{m}\mid \Varid{m}\to \Varid{r}\;\mathbf{where}{}\<[E]%
\\
\>[B]{}\hsindent{5}{}\<[5]%
\>[5]{}\Varid{ask}{}\<[15]%
\>[15]{}\mathbin{::}{}\<[15E]%
\>[19]{}\Varid{m}\;\Varid{r}{}\<[E]%
\\
\>[B]{}\hsindent{5}{}\<[5]%
\>[5]{}\Varid{ask}{}\<[15]%
\>[15]{}\mathrel{=}{}\<[15E]%
\>[19]{}\Varid{reader}\;\Varid{id}{}\<[E]%
\\[\blanklineskip]%
\>[B]{}\hsindent{5}{}\<[5]%
\>[5]{}\Varid{local}{}\<[15]%
\>[15]{}\mathbin{::}{}\<[15E]%
\>[19]{}(\Varid{r}\to \Varid{r})\to \Varid{m}\;\Varid{a}\to \Varid{m}\;\Varid{a}{}\<[E]%
\\[\blanklineskip]%
\>[B]{}\hsindent{5}{}\<[5]%
\>[5]{}\Varid{reader}{}\<[15]%
\>[15]{}\mathbin{::}{}\<[15E]%
\>[19]{}(\Varid{r}\to \Varid{a})\to \Varid{m}\;\Varid{a}{}\<[E]%
\\
\>[B]{}\hsindent{5}{}\<[5]%
\>[5]{}\Varid{reader}\;\Varid{f}{}\<[15]%
\>[15]{}\mathrel{=}{}\<[15E]%
\>[19]{}\mathbf{do}\;{}\<[23]%
\>[23]{}\{\mskip1.5mu \Varid{r}\leftarrow \Varid{ask};\Varid{return}\;(\Varid{f}\;\Varid{r});\mskip1.5mu\}{}\<[E]%
\ColumnHook
\end{hscode}\resethooks

The documentation of the source code specifies that the method \ensuremath{\Varid{ask}}
should return the environment of the computation. We assume that 
\ensuremath{\Varid{ask}} has no side effect (though the documentation leaves it unspecified).
The method \ensuremath{\Varid{reader}} generalizes it by allowing to
apply an arbitrary function to the environment. Up to isomorphism, 
this method embeds the reader monad (defined in Sect.~\ref{mtl}) 
into the monad~\ensuremath{\Varid{m}}. This method is analogous to $\rho$ of 
Sect.~\ref{except} that similarly embeds the exception monad into~\ensuremath{\Varid{m}}. 
In the hierarchy described in Sect.~\ref{except}, both \ensuremath{\Varid{ask}} and \ensuremath{\Varid{reader}} 
belong to the first level as point functions. The method \ensuremath{\Varid{local}} 
sets up a modified environment for a local computation. It is
a functor application by nature 
but the Haskell type does not reflect this because 
the environment type is not a parameter of~\ensuremath{\Varid{m}}. 

We abandon some Haskell function names in favour of mathematical notation. 
Similarly to what we saw in the case of exceptions, we here treat the 
environment type as an extra (first) parameter of the monad and denote 
the obtained bifunctor by $\frf$. Hence $\ensuremath{\Varid{local}}\ap h$ is written as 
$\frf\ap(h,\id)$. Note that $\frf$ is contravariant in its first argument, 
i.e., if $h\oftyp\obr'\to\obr$ then
$\frf\ap(h,\id)\oftyp\frf\ap(\obr,\obx)\to\frf\ap(\obr',\obx)$. 
The method \ensuremath{\Varid{reader}} is denoted by
$\rho\oftyp(\obr\to\obx)\to\frf\ap(\obr,\obx)$. The arrow that
constructs function spaces (like in $\obr\to\obx$) can be made a bifunctor
by defining, for any $h\oftyp\obr'\to\obr$ and $f\oftyp\obx\to\obx'$, 
a new function $h\to f\oftyp(\obr\to\obx)\to(\obr'\to\obx')$ by the equation
$(h\to f)\ap g=f\icp g\icp h$. We will use this notation occasionally 
in this section.

In Subsection~\ref{reader:reader}, we propose an axiomatization of 
computations in environment which directly implies its every model being 
isomorphic to a monad obtained by an application of the reader monad transformer.
It turns out that the exception, reader, writer and state monad transformers
preserve the axioms. This axiomatics uses $\rho$ as a primitive; 
in Subsect.~\ref{reader:ask}, 
we consider an equivalent axiomatics that uses $\ask$ as a primitive and
defines $\rho$ in terms of it.

\subsection{Reduction to Reader Transformer
Applications}\label{reader:reader}

Recall that we assume computations in environment being described by a 
bifunctor~$\frf$ whose first parameter is the environment type and 
second parameter is the type of the return value. The functor laws are
$\frf\ap(\id,\id)=\id$ and 
$\frf\ap(h\icp h',f'\icp f)=\frf\ap(h',f')\icp\frf\ap(h,f)$; note the change
in the order of the composed morphisms in the first argument due to
contravariance. As in the case of exceptions, denote the monad unit and bind
by $\invok$ and $\hole\bnd$, respectively; their types here are 
$\invok\oftyp\obx\to\frf\ap(\obr,\obx)$ and 
$\hole\bnd\oftyp(\obx\to\frf\ap(\obr,\obx'))\to\frf\ap(\obr,\obx)\to\frf\ap(\obr,\obx')$. 

Developing an axiomatics that would imply its models being isomorphic to
reader transformer applications requires introducing operations that have no
counterpart in Haskell MTL. Before doing it, consider the laws to be required 
that are expressible in terms of standard operations. Firstly, changing
the environment by $\frf\ap(h,\id)$ being a monad morphism:
\lawarray{
&\frf\ap(h,\id)\icp\invok
&&\al{=}\invok
\label{reader:bifun-unithom}\sctag{Bifun-UnitHom}\\
&\frf\ap(h,\id)\icp k\bnd
&&\al{=}(\frf\ap(h,\id)\icp k)\bnd\icp\frf\ap(h,\id)
\label{reader:bifun-bndhom}\sctag{Bifun-BndHom}
}
Next, $\rho$ being a monad morphism from the underlying 
reader monad $\obr\to\_$ to the monad $\frf\ap(\obr,\_)$:
\lawarray{
&\rho\icp\const
&&\al{=}\invok
\label{reader:rdr-unithom}\sctag{Rdr-UnitHom}\\
&\rho\icp k^*
&&\al{=}(\rho\icp k)\bnd\icp\rho
\label{reader:rdr-bndhom}\sctag{Rdr-BndHom}
}
Here, $\hole^*$ denotes the bind operation of the underlying reader monad; 
note that $\const$ is its unit. And lastly, $\rho$ being a natural 
transformation between the power bifunctor and~$\frf$:
\lawarray{
&\rho\icp(h\to f)
&&\al{=}\frf\ap(h,f)\icp\rho
\label{reader:rdr-nat}\sctag{Rdr-Nat}
}

Although valid in all monads considered in this paper, these axioms are not
as powerful as we could do by widening our point of view. 
The underlying assumption of our approach is that 
types of the form $\frf\ap(\obr,\obx)$ encode, in some way, 
environment-dependent monadic computations. The dependency does not have to 
occur in the form of a function whose argument type is~$\obr$,
because applying, for instance, the state monad transformer with state 
type~$\obs$ to a member of \ensuremath{\Conid{MonadReader}} with the environment
type~$\obr$ produces functions with argument type~$\obs$ (i.e., not $\obr$) 
inheriting the dependency on~$\obr$ from the member of \ensuremath{\Conid{MonadReader}}. 
Therefore, we introduce functions $\apply$ and $\abstr$ 
establishing an isomorphism between types $\frf\ap(\obr,\obx)$ and
$\obr\to\frm\ap\obx$ for a monad~$\frm$. More precisely:
\lawarray{
&\apply
&&\al{\oftyp}\frf\ap(\obr,\obx)\to\obr\to\frm\ap\obx\\
&\abstr
&&\al{\oftyp}(\obr\to\frm\ap\obx)\to\frf\ap(\obr,\obx)
}
So a computation of the form
$\applyof t\ap r$ fixes the environment of the 
environment-dependent computation~$t$ to be $r$, which intuitively is 
an application of a hidden function, 
and $\abstrof f$ \qquot{abstracts} the parameter of
its argument function~$f$ by hiding it into the functor computation. 

Table~\ref{reader:def} presents precise definitions of the exception,
reader, writer and state monad transformers (denoted by $\ExceptTof{\obe}$,
$\ReaderTof{\obq}$, $\WriterTof{\obw}$ and $\StateTof{\obs}$, respectively)
along with the corresponding bifunctor transformers (which are denoted
similarly) and 
specifies propagation of $\rho$, $\apply$ and $\abstr$ through the 
transformers; we omit definitions of morphism mappings of the functors and 
monad operations as they are standard. (In the defining equations for
$\rho$, $\apply$ and $\abstr$, the occurrences of these operations in the
l.h.s. are those of the bifunctor constructed by the transformer 
while the occurrences in the r.h.s. belong to the original bifunctor~$\frf$.)  
By sequential application of these transformers in all possible orders, 
we achieve a set of infinitely many structures each either 
having $\frf$, $\frm$ and the related operations defined in
terms of those of a simpler member structure of this set or being a base
case, for which we can take $\frf\ap(\obr,\obx)=\obr\to\frm\ap\obx$ with 
$\rho=\id\to\invok$ and both $\apply$ and $\abstr$ defined as identities. 
Note that propagation of~$\rho$ 
is for all transformers defined via composing from the left with the lift operation 
of the particular transformer; this matches the definition of \ensuremath{\Varid{reader}} in MTL
in the case of exception, writer, and state monad transformer. 

\begin{table*}
\caption{Definitions of functions $\rho$, $\apply$ and $\abstr$ for
different monad transformers.}\label{reader:def}
\[
\begin{array}{c}
\begin{array}{rcl}
\multicolumn{3}{@{}l}{\mbox{Exception transformer:}}\\
\ExceptTof{\obe}\frm\ap\obx&=&\frm\ap(\obe+\obx)\\
\ExceptTof{\obe}\frf\ap(\obr,\obx)&=&\frf\ap(\obr,\obe+\obx)\\
\rho&=&\frf\ap(\id,\inr)\icp\rho\\
\applyof t\ap r&=&\applyof t\ap r\\
\abstrof f&=&\abstrof f
\end{array}\;\;\;
\begin{array}{rcl}
\multicolumn{3}{@{}l}{\mbox{Reader transformer:}}\\
\ReaderTof{\obq}\frm\ap\obx&=&\obq\to\frm\ap\obx\\
\ReaderTof{\obq}\frf\ap(\obr,\obx)&=&\obq\to\frf\ap(\obr,\obx)\\
\rho&=&\const\icp\rho\\
\applyof t\ap r&=&\lam{q}{\applyof(t\ap q)\ap r}\\
\abstrof f&=&\lam{q}{\abstrof(\lam{r}{f\ap r\ap q})}
\end{array}\;
\begin{array}{rcl}
\multicolumn{3}{@{}l}{\mbox{Writer transformer:}}\\
\WriterTof{\obw}\frm\ap\obx&=&\frm\ap(\obx\times\obw)\\
\WriterTof{\obw}\frf\ap(\obr,\obx)&=&\frf\ap(\obr,\obx\times\obw)\\
\rho&=&\frf\ap(\id,\id\both\constof\one)\icp\rho\\
\applyof t\ap r&=&\applyof t\ap r\\
\abstrof f&=&\abstrof f
\end{array}\\\\
\begin{array}{rcl}
\multicolumn{3}{@{}l}{\mbox{State transformer:}}\\
\StateTof{\obs}\frm\ap\obx&=&\obs\to\frm\ap(\obx\times\obs)\\
\StateTof{\obs}\frf\ap(\obr,\obx)&=&\obs\to\frf\ap(\obr,\obx\times\obs)\\
\rho&=&(\lam{ts}{\frf\ap(\id,\id\both\constof s)\ap t})\icp\rho\\
\applyof t\ap r&=&\lam{s}{\applyof(t\ap s)\ap r}\\
\abstrof f&=&\lam{s}{\abstrof(\lam{r}{f\ap r\ap s})}
\end{array}\\\hline
\end{array}
\]
\end{table*}

Denote the unit and bind of~$\frm$ by
$\invok$ and $\hole\bnd$ like those of $\frf\ap(\obr,\_)$. 
We specify $\apply$ equationally as a function translating 
the operations of~$\frf$ to operations of~$\frm$:
\lawarray{
&\apply\icp\frf\ap(h,f)
&&\al{=}(h\to\frm\ap f)\icp\apply
\label{reader:apply-nat}\sctag{App-Nat}\\
&\apply\icp\invok
&&\al{=}\const\icp\invok
\label{reader:apply-unithom}\sctag{App-UnitHom}\\
&\apply\icp k\bnd
&&\al{=}(\apply\icp k)\lbnd\icp\apply
\label{reader:apply-bndhom}\sctag{App-BndHom}\\
&\apply\icp\rho
&&\al{=}\id\to\invok
\label{reader:apply-rdr}\sctag{App-Rdr}
}
Here, $\hole\lbnd$ in the r.h.s. of \ref{reader:apply-bndhom} 
denotes the monad bind of~$\frm$ 
lifted to functions. Note that $\const\icp\invok$
in \ref{reader:apply-unithom} similarly lifts $\invok$ to functions. 
We also assume that $\apply$ and $\abstr$ are inverses of each
other:
\lawarray{
&\apply\icp\abstr
&&\al{=}\id
\label{reader:apply-abstr}\sctag{App-Abs}\\
&\abstr\icp\apply
&&\al{=}\id
\label{reader:abstr-apply}\sctag{Abs-App}
}
As a consequence, proving properties of $\frf$ and $\rho$ are reduced to
proving properties of monad $\frm$. The following theorems hold in $\Set$:

\begin{theorem}\label{reader:conc}
Let $\frm$ be a monad. Let $\frf$ be a type-preserving, contravariant in its
first argument, binary mapping of objects and morphisms. Let $\rho$,
$\apply$, $\abstr$ have their right types and meet axioms
\ref{reader:apply-nat}, \ref{reader:apply-unithom},
\ref{reader:apply-bndhom}, \ref{reader:apply-rdr}, \ref{reader:apply-abstr}
and \ref{reader:abstr-apply}. Then $\frf$ meets the functor laws, its
left section for every type $\obr$ is a monad w.r.t. $\invok$ and
$\hole\bnd$, and the equations \ref{reader:bifun-unithom},
\ref{reader:bifun-bndhom}, \ref{reader:rdr-unithom}, \ref{reader:rdr-bndhom}
and \ref{reader:rdr-nat} are all valid.
\end{theorem}

\begin{theorem}\label{reader:main}
Let $\frm$ be a monad. Then: 
\begin{itemize}
\item The bifunctor obtained by applying the reader monad
transformer with an environment type~$\obr$ to $\frm$, with 
$\apply$ and $\abstr$ defined as identities and other operations 
defined like in the Haskell MTL,
satisfies the axioms \ref{reader:apply-nat},
\ref{reader:apply-unithom}, \ref{reader:apply-bndhom}, \ref{reader:apply-rdr}, 
\ref{reader:apply-abstr}
and \ref{reader:abstr-apply}; 
\item Applying the exception, reader, writer and state bifunctor (and monad)
transformers 
preserve these axioms.
\end{itemize}
\end{theorem}

Proofs are straightforward.

\subsection{Axioms of \ensuremath{\Varid{ask}}}\label{reader:ask}

The definition of class \ensuremath{\Conid{MonadReader}} provides mutual definitions of
\ensuremath{\Varid{reader}} and \ensuremath{\Varid{ask}}, looking in our language as follows:
\lawarray{
&\ask
&&\al{=}\rho(\id)
\label{reader:ask-rdr}\sctag{Ask-Rdr}\\
&\rho(f)
&&\al{=}\frf\ap(\id,f)\ap\ask
\label{reader:rdr-ask}\sctag{Rdr-Ask}
}
In order to axiomatize $\ask$ instead of $\rho$, we replace
\ref{reader:apply-rdr} with
\lawarray{
&\applyof\ask
&&\al{=}\invok
\label{reader:apply-ask}\sctag{App-Ask}
}
That is, asking of the environment as a computation of type
$\frf\ap(\obr,\obr)$, when presented as a function of type
$\obr\to\frm\ap\obr$, equals the monad unit of $\frm$ that immediately 
returns the argument environment.

This gives an equivalent axiomatics indeed, as established by the following 
theorem:

\begin{theorem}\label{reader:askthm}
Let $\frm$ be a monad. Let $\frf$ be a type-preserving, contravariant in its
first argument, binary mapping of objects and morphisms. Let $\rho$,
$\apply$, $\abstr$ have their right types and meet axioms
\ref{reader:apply-nat}, \ref{reader:apply-unithom},
\ref{reader:apply-bndhom}, \ref{reader:apply-abstr} 
and \ref{reader:abstr-apply}. Then the set of equations 
$\set{\ref{reader:apply-rdr},\ref{reader:ask-rdr}}$ 
is equivalent to the set of equations 
$\set{\ref{reader:apply-ask},\ref{reader:rdr-ask}}$.
\end{theorem}

\begin{proof}
Straightforward but we present it fully 
for illustrating equational reasoning in
our axiomatics. If we assume \ref{reader:apply-rdr}, the premises of
Theorem~\ref{reader:conc} are fully met which allows us to use 
\ref{reader:rdr-nat} in the proof of \ref{reader:rdr-ask}. 
In the proof of \ref{reader:ask-rdr}, we can use functor laws since their
proof does not need \ref{reader:apply-rdr}. 

\begin{itemize}
\item
$\set{\ref{reader:apply-rdr},\ref{reader:ask-rdr}}\Longrightarrow\set{\ref{reader:apply-ask},\ref{reader:rdr-ask}}$:\vspace{0.5ex}

\begin{center}
\begin{minipage}{0.4\columnwidth}
\ref{reader:apply-ask}:
\prfarray{
&\applyof\ask\\
\al{=}&\since{\ref{reader:ask-rdr}}\\
&\applyof(\rho(\id))\\
\al{=}&\since{\ref{reader:apply-rdr}}\\
&(\id\to\invok)\ap\id\\
\al{=}&\since{power, identity}\\
&\invok
}
\end{minipage}\qquad
\begin{minipage}{0.4\columnwidth}
\ref{reader:rdr-ask}:
\prfarray{
&\frf\ap(\id,f)\ap\ask\\
\al{=}&\since{\ref{reader:ask-rdr}}\\
&\frf\ap(\id,f)\ap(\rho(\id))\\
\al{=}&\since{\ref{reader:rdr-nat}}\\
&\rho((\id\to f)\ap\id)\\
\al{=}&\since{power, identity}\\
&\rho(f)
}
\end{minipage}
\end{center}

\item
$\set{\ref{reader:apply-ask},\ref{reader:rdr-ask}}\Longrightarrow\set{\ref{reader:apply-rdr},\ref{reader:ask-rdr}}$:

\begin{center}
\begin{minipage}[t]{0.45\columnwidth}
\ref{reader:apply-rdr}:
\prfarray{
&\applyof(\rho(f))\\
\al{=}&\since{\ref{reader:rdr-ask}}\\
&\applyof(\frf\ap(\id,f)\ap\ask)\\
\al{=}&\since{\ref{reader:apply-nat}}\\
&(\id\!\!\to\!\!\frm\ap f)\ap(\applyof\ask)\\
\al{=}&\since{\ref{reader:apply-ask}}\\
&(\id\!\!\to\!\!\frm\ap f)\ap\invok\\
\al{=}&\since{power, identity}\\
&\frm\ap f\icp\invok\\
\al{=}&\since{naturality}\\
&\invok\icp f\\
\al{=}&\since{power, identity}\\
&(\id\to\invok)\ap f
}
\end{minipage}\quad
\begin{minipage}[t]{0.4\columnwidth}
\ref{reader:ask-rdr}:
\prfarray{
&\rho(\id)\\
\al{=}&\since{\ref{reader:rdr-ask}}\\
&\frf\ap(\id,\id)\ap\ask\\
\al{=}&\since{functor, identity}\\
&\ask
}
\end{minipage}
\end{center}
\end{itemize}

\end{proof}

\section{An Abstract View and its Application to \ensuremath{\Conid{MonadWriter}}}\label{writer}

The class \ensuremath{\Conid{MonadWriter}} is defined in MTL as follows:
\begin{hscode}\SaveRestoreHook
\column{B}{@{}>{\hspre}l<{\hspost}@{}}%
\column{5}{@{}>{\hspre}l<{\hspost}@{}}%
\column{13}{@{}>{\hspre}l<{\hspost}@{}}%
\column{23}{@{}>{\hspre}l<{\hspost}@{}}%
\column{28}{@{}>{\hspre}c<{\hspost}@{}}%
\column{28E}{@{}l@{}}%
\column{32}{@{}>{\hspre}l<{\hspost}@{}}%
\column{E}{@{}>{\hspre}l<{\hspost}@{}}%
\>[B]{}\mathbf{class}\;(\Conid{Monoid}\;\Varid{w},\Conid{Monad}\;\Varid{m}){}\<[28]%
\>[28]{}\Rightarrow {}\<[28E]%
\>[32]{}\Conid{MonadWriter}\;\Varid{w}\;\Varid{m}{}\<[E]%
\\
\>[28]{}\mid {}\<[28E]%
\>[32]{}\Varid{m}\to \Varid{w}\;\mathbf{where}{}\<[E]%
\\
\>[B]{}\hsindent{5}{}\<[5]%
\>[5]{}\Varid{writer}{}\<[13]%
\>[13]{}\mathbin{::}(\Varid{a},\Varid{w})\to \Varid{m}\;\Varid{a}{}\<[E]%
\\
\>[B]{}\hsindent{5}{}\<[5]%
\>[5]{}\Varid{tell}{}\<[13]%
\>[13]{}\mathbin{::}\Varid{w}\to \Varid{m}\;(){}\<[E]%
\\
\>[B]{}\hsindent{5}{}\<[5]%
\>[5]{}\Varid{listen}{}\<[13]%
\>[13]{}\mathbin{::}\Varid{m}\;\Varid{a}\to \Varid{m}\;(\Varid{a},\Varid{w}){}\<[E]%
\\
\>[B]{}\hsindent{5}{}\<[5]%
\>[5]{}\Varid{pass}{}\<[13]%
\>[13]{}\mathbin{::}\Varid{m}\;(\Varid{a},\Varid{w}\to \Varid{w})\to \Varid{m}\;\Varid{a}{}\<[E]%
\\[\blanklineskip]%
\>[B]{}\hsindent{5}{}\<[5]%
\>[5]{}\Varid{writer}{}\<[13]%
\>[13]{}\mathord{\sim}(\Varid{a},\Varid{w}){}\<[23]%
\>[23]{}\mathrel{=}\mathbf{do}\;\{\mskip1.5mu \Varid{tell}\;\Varid{w};\Varid{return}\;\Varid{a};\mskip1.5mu\}{}\<[E]%
\\
\>[B]{}\hsindent{5}{}\<[5]%
\>[5]{}\Varid{tell}\;{}\<[13]%
\>[13]{}\Varid{w}{}\<[23]%
\>[23]{}\mathrel{=}\Varid{writer}\;((),\Varid{w}){}\<[E]%
\ColumnHook
\end{hscode}\resethooks
The method \ensuremath{\Varid{writer}}, similarly to the method \ensuremath{\Varid{reader}} in the class 
\ensuremath{\Conid{MonadReader}}, embeds the writer monad into the monad~\ensuremath{\Varid{m}}. 
The method \ensuremath{\Varid{tell}} logs the value given as argument and immediately 
returns the value~\ensuremath{()}. It is a special case of \ensuremath{\Varid{writer}}. Both belong to the
first level in the hierarchy of Sect.~\ref{except}.

The method \ensuremath{\Varid{listen}} applies to a monadic computation and copies the whole
log of this computation into its return value. The method 
\ensuremath{\Varid{pass}} applies to a monadic computation, the return value of which
contains a function, and modifies the log of this computation by applying
this function. Note that both \ensuremath{\Varid{listen}} and \ensuremath{\Varid{pass}}
modify a monadic computation in a way that can be encoded as a
transformation of values of the writer monad, i.e., in the form of
a function of type $\obx\times\obw\to\obx'\times\obw$: for \ensuremath{\Varid{listen}},
the corresponding function is $\lam{(a,w)}{((a,w),w)}$, and for \ensuremath{\Varid{pass}}, the function
is $\lam{((a,f),w)}{(a,f\ap w)}$. According to
Sect.~\ref{except}, these methods are mixmaps and belong to the second level
of the method hierarchy. 
If the class \ensuremath{\Conid{MonadWriter}} contained a general method 
\ensuremath{\Varid{mixmap}\mathbin{::}((\Varid{a},\Varid{w})\to (\Varid{a'},\Varid{w}))\to \Varid{m}\;\Varid{a}\to \Varid{m}\;\Varid{a'}}, we could define
\ensuremath{\Varid{listen}} and \ensuremath{\Varid{pass}} as
\begin{hscode}\SaveRestoreHook
\column{B}{@{}>{\hspre}l<{\hspost}@{}}%
\column{9}{@{}>{\hspre}l<{\hspost}@{}}%
\column{E}{@{}>{\hspre}l<{\hspost}@{}}%
\>[B]{}\Varid{listen}{}\<[9]%
\>[9]{}\mathrel{=}\Varid{mixmap}\;(\lambda (\Varid{a},\Varid{w})\to ((\Varid{a},\Varid{w}),\Varid{w})){}\<[E]%
\\
\>[B]{}\Varid{pass}{}\<[9]%
\>[9]{}\mathrel{=}\Varid{mixmap}\;(\lambda ((\Varid{a},\Varid{f}),\Varid{w})\to (\Varid{a},\Varid{f}\;\Varid{w})){}\<[E]%
\ColumnHook
\end{hscode}\resethooks
On the other hand, 
\ensuremath{\Varid{mixmap}} can be defined in terms of \ensuremath{\Varid{listen}} and \ensuremath{\Varid{pass}} by
\begin{hscode}\SaveRestoreHook
\column{B}{@{}>{\hspre}l<{\hspost}@{}}%
\column{E}{@{}>{\hspre}l<{\hspost}@{}}%
\>[B]{}\Varid{mixmap}\;\Varid{g}\mathrel{=}\Varid{pass}\mathbin{\circ}\Varid{fmap}\;(\Varid{bimap}\;\Varid{id}\;\Varid{const}\mathbin{\circ}\Varid{g})\mathbin{\circ}\Varid{listen}{}\<[E]%
\ColumnHook
\end{hscode}\resethooks

Unlike in the case of exceptions and environments, we do not consider the
monoid parameter of the class as an extra parameter of the monad. 
Keeping the monoid fixed enables us to derive the structure of the
\ensuremath{\Conid{WriterMonad}} class methods as some kind of reflection of 
the structure of the underlying writer monad. This is not to say that 
generalizing the approach by letting the monoid vary would be pointless. 
In the rest of this section, we denote the monad under consideration by~$\frm$
but also use the bifunctor notation occasionally in
Subsections~\ref{writer:bind} and~\ref{writer:other} (with fixed monoid).

\subsection{Pointed Functors with Mixmap}\label{writer:mixmap}

Analogously to the case of exceptions, 
denote the relative point and mixmap by $\rho$ and $\phi$, respectively. 
The following set of equations is valid for all monads
of the class \ensuremath{\Conid{MonadWriter}} that have been constructed by sequential
application of the exception, reader, writer and state monad transformers:
\lawarray{
&\rho\icp f
&&\al{=}\phi(f)\icp\rho
\label{writer:rpoint-rnat}\sctag{RPoint-RNat}\\
&\phi(\id)
&&\al{=}\id
\label{writer:rfun-id}\sctag{Mixmap-Id}\\
&\phi(g\icp f)
&&\al{=}\phi(g)\icp\phi(f)
\label{writer:rfun-comp}\sctag{Mixmap-Comp}
}
The first equation is a \qquot{relative naturality} 
law that uses mixmap as one of the two functors that $\rho$ is working
between (the other one is identity). The other two laws
establish preservation of identity and composition. The last
one is particularly interesting because, as shown
in~\cite{DBLP:conf/ictac/Nestra19}, it does not hold in the
axiomatics of exceptions we saw in Sect.~\ref{except}. 
In this sense, \ensuremath{\Conid{MonadWriter}} behaves more nicely than \ensuremath{\Conid{MonadError}}.

For finding properties of \ensuremath{\Varid{tell}}, 
\ensuremath{\Varid{listen}} and \ensuremath{\Varid{pass}}, one can rely on 
the mixmap laws and the previously seen definitions of these methods 
in terms of mixmap.
We will not search for an equivalent axiomatics using these methods as
primitives but we define two simpler functions and find an
equivalent axiomatics for them. To this end, note that \ensuremath{\Varid{listen}}
and \ensuremath{\Varid{pass}} serve dual purposes in the sense that \ensuremath{\Varid{listen}} copies
the log into the return value while \ensuremath{\Varid{pass}} moves information from the
return value to the log. Instead of \ensuremath{\Varid{listen}} and \ensuremath{\Varid{pass}}, we consider
$\shift\oftyp\frm\ap\obx\to\frm\ap(\obx\times\obw)$ and 
$\fuse\oftyp\frm\ap(\obx\times\obw)\to\frm\ap\obx$, defined by equations
$\shift=\phi(\lam{(a,w)}{((a,w),\one)})$ and 
$\fuse=\phi(\lam{((a,w'),w)}{(a,w\cdot w')})$, that achieve the 
same aims in a cleaner way. 
The function $\shift$ copies the log into the return value but, unlike
\ensuremath{\Varid{listen}}, replaces the original with the monoid unit. The function
$\fuse$ uses the monoid multiplication to join a monoid element in the return
value with the current log. In terms of $\shift$ and $\fuse$, the general 
mixmap can be expressed as
\lawarray{
&\phi(g)
&&\al{=}\fuse\icp\frm\ap g\icp\shift
\label{writer:rfun-fuseshift}\sctag{Mixmap-FuseShift}
}

Note that $\lam{p}{(p,\one)}$ and 
$\lam{((a,w'),w)}{(a,w\cdot w')}$ are the unit and join, respectively, of
the writer monad. Denoting the unit and join by
$\eta$ and $\mu$, respectively, we can abstract from the underlying monad and
specify $\shift$ and $\fuse$ by
\lawarray{
&\shift
&&\al{=}\phi(\eta)
\label{writer:shift-rfun}\sctag{Shift-Mixmap}\\
&\fuse
&&\al{=}\phi(\mu)
\label{writer:fuse-rfun}\sctag{Fuse-Mixmap}
}
The types in the abstract view are 
$\shift\oftyp\frm\ap\obx\to\frm\ap(\frj\ap\obx)$ and
$\fuse\oftyp\frm\ap(\frj\ap\obx)\to\frm\ap\obx$ where 
$\frj$ denotes the underlying monad. 
The mapping of morphisms by $\frm$ can be given via~$\phi$:
\lawarray{
&\frm\ap f
&&\al{=}\phi(\frj\ap f)
\label{writer:fun-rfun}\sctag{Fun-Mixmap}
}
Assuming this as definition, one
can prove the two functor laws for $\frm$ using \ref{writer:rfun-id},
\ref{writer:rfun-comp} and the functor laws for $\frj$. Moreover, the unit
of $\frm$ can be expressed by
\lawarray{
&\invok
&&\al{=}\rho\icp\eta
\label{writer:point-rpoint}\sctag{Point-RPoint}
}
Proof of naturality of $\invok$ is straightforward 
using naturality of $\eta$ along with \ref{writer:rpoint-rnat} and
\ref{writer:fun-rfun}.

We are now going to form an alternative axiom set that uses 
$\shift$ and $\fuse$ instead of $\phi$ as the underlying operations. 
Firstly, we include the following laws that resemble 
the coherence conditions of monad join:
\lawarray{
&\fuse\icp\frm\ap\mu
&&\al{=}\fuse\icp\fuse
\label{writer:fuse-fuse}\sctag{Fuse-Fuse}\\
&\fuse\icp\frm\ap\eta
&&\al{=}\id
\label{writer:fuse-funpoint}\sctag{Fuse-FunPoint}\\
&\fuse\icp\shift
&&\al{=}\id
\label{writer:fuse-shift}\sctag{Fuse-Shift}
}
Secondly, we include the following two homomorphism laws for $\shift$:
\lawarray{
&\shift\icp\rho
&&\al{=}\invok
\label{writer:shift-rpointhom}\sctag{Shift-PointHom}\\
&\shift\icp\phi(g)
&&\al{=}\frm\ap g\icp\shift
\label{writer:shift-rfunhom}\sctag{Shift-RNat}
}
Note that \ref{writer:shift-rpointhom} and \ref{writer:shift-rfunhom}
uniquely determine $\rho$ and $\phi$. Hence there is no need to include
direct definitions of $\rho$ and $\phi$; but if we did it, $\rho$ would be
given by
\lawarray{
&\rho
&&\al{=}\fuse\icp\invok
\label{writer:rpoint-point}\sctag{RPoint-Point}
}
and $\phi$ by \ref{writer:rfun-fuseshift} (for proof, compose
\ref{writer:shift-rpointhom} and \ref{writer:shift-rfunhom} with $\fuse$ from 
the left and apply \ref{writer:fuse-shift}). Then we could replace 
\ref{writer:shift-rpointhom} and \ref{writer:shift-rfunhom} with 
axioms not mentioning $\rho$ and $\phi$ to obtain an 
axiomatization expressed fully in terms of $\shift$ and $\fuse$. We prefer 
the homomorphism laws for brevity and elegance. 

Altogether, we have the
following result that establishes equivalence of the axiomatics of~$\phi$ and 
the axiomatics of $\shift$ and $\fuse$ in every category:

\begin{theorem}\label{writer:level2}
Let $(\frj,\eta,\mu)$ be a monad. Suppose that $\frm$ consists of 
a mapping of objects to objects and a type-preserving mapping of 
morphisms to morphisms. (This is to say that $\frm$ is an endofunctor
without assuming functor laws.) Furthermore, assume transformations 
$\rho\oftyp\frj\ap\obx\to\frm\ap\obx$, $\invok\oftyp\obx\to\frm\ap\obx$, 
$\shift\oftyp\frm\ap\obx\to\frm\ap(\frj\ap\obx)$ and
$\fuse\oftyp\frm\ap(\frj\ap\obx)\to\frm\ap\obx$ along with $\phi$ that maps 
morphisms of type $\frj\ap\obx\to\frj\ap\obx'$ to morphisms of type
$\frm\ap\obx\to\frm\ap\obx'$ being given. Then the set of laws consisting of
\ref{writer:rpoint-rnat}, \ref{writer:rfun-id}, \ref{writer:rfun-comp},
\ref{writer:shift-rfun}, \ref{writer:fuse-rfun}, \ref{writer:fun-rfun} and 
\ref{writer:point-rpoint} is equivalent to the set of laws consisting of:
\begin{itemize}
\item Two functor laws for $\frm$;
\item Naturality of $\invok$ and $\fuse$;
\item The coherence laws \ref{writer:fuse-fuse}, \ref{writer:fuse-funpoint},
and \ref{writer:fuse-shift};
\item The homomorphism laws \ref{writer:shift-rpointhom} and
\ref{writer:shift-rfunhom}.
\end{itemize}
\end{theorem}

The proofs are straightforward.

We have not mentioned naturality of $\rho$ and $\shift$; both are
implied by the axioms considered. One can also deduce the following 
dual of \ref{writer:fuse-fuse} law from the axioms:
\lawarray{
&\frm\ap\eta\icp\shift
&&\al{=}\shift\icp\shift
\label{writer:shift-shift}\sctag{Shift-Shift}
}

\subsection{Two-Story Monads}\label{writer:bind}

So far, bind operation of $\frm$ was not involved into our study. At the
third level of the hierarchy defined in Sect.~\ref{except} for exceptions, 
we also had a \qquot{joint handle} function of type
$(\obe+\oba\to\frf\ap(\obe',\oba'))\to\frf\ap(\obe,\oba)\to\frf\ap(\obe',\oba')$
generalizing the monad bind; we noted that its type equals the type of bind of 
a relative monad on $+$ but it unfortunately does not meet all relative
monad laws. One can also define a similar function 
of type $(\obx\times\obw\to\frm\ap\obx')\to\obm\ap\obx\to\obm\ap\obx'$ 
that takes the current log along with the return value into account when
binding two computations with writer effects. All
relative monad laws turn out to be
satisfied for all structures constructed via
the four monad transformers we consider in this paper. 
Therefore we start axiomatizing of the third level from the relative monad
laws (with, again, $\frj$ replacing the writer monad):
\lawarray{
&k\hdl\icp\rho
&&\al{=}k
\label{writer:hdl-unitl}\sctag{RBnd-UnitL}\\
&\rho\hdl
&&\al{=}\id
\label{writer:hdl-id}\sctag{RBnd-Id}\\
&l\hdl\icp k\hdl
&&\al{=}(l\hdl\icp k)\hdl
\label{writer:hdl-assoc}\sctag{RBnd-Assoc}\\
&\frm\ap f
&&\al{=}(\rho\icp\frj\ap f)\hdl
\label{writer:fun-hdl}\sctag{Fun-RBnd}
}
In order to be able to express the usual monad bind of type 
$(\obx\to\frm\ap\obx')\to\frm\ap\obx\to\frm\ap\obx'$ 
in terms of the relative monad bind whose type in the case of 
an abstract base functor~$\frj$ is 
$(\frj\ap\obx\to\frm\ap\obx')\to\frm\ap\obx\to\frm\ap\obx'$, we consider a 
pseudobind $\hole\pbnd$ of type
$(\obx\to\frm\ap\obx')\to\frj\ap\obx\to\frm\ap\obx'$. Then one can define
bind of $\frm$, denoted by $\hole\bnd$, by
\lawarray{
&k\bnd
&&\al{=}k\pbnd\hdl
\label{writer:bnd-pbndhdl}\sctag{Bnd-PBndRBnd}
}
(We use the half-star notation for bind of $\frm$ and leave the standard
notation $\hole^*$ for bind of the monad $\frj$. The term \emph{pseudobind}
was chosen after Steele~Jr.~\cite{DBLP:conf/popl/Steele94}; 
we will discuss Steele's work in
Subsect.~\ref{writer:steele}.) 
In the writer case, we can take
\lawarray{
&k\pbnd
&&\al{=}\lam{(a,w)}{\frf\ap(\rsect{w}{\cdot},\id)\ap(k\ap a)}\mbox{,}
\label{writer:pbnd-bifun}\sctag{PBnd-Bifun}
}
where 
the section notation of Haskell is used in the first argument of $\frf$,
which itself is the bifunctor obtained from $\frm$ by treating $\obw$ as its 
(first) parameter. (The type $\obw$ is still fixed. The domain of the
additional parameter of $\frf$ is the category consisting of a singleton
object~$\obw$ and its endomorphisms.)
So $k\pbnd$ takes a pair $(a,w)$, applies $k$ to $a$ and multiplies the log
of the computation by $w$ from the left.
We could not use $\phi(\id\times\rsect{w}{\cdot})$ instead as 
the exception monad transformer does not preserve the equality
$\frf(f,\id)=\phi(\id\times f)$. We will study~$\frf$ more in
Subsect.~\ref{writer:other}.

For $\hole\pbnd$, we use the following monad-like axioms:
\lawarray{
&k\pbnd\icp\eta
&&\al{=}k
\label{writer:pbnd-unitl}\sctag{PBnd-UnitL}\\
&(\rho\icp f)\pbnd
&&\al{=}\rho\icp f^*
\label{writer:pbnd-unitr}\sctag{PBnd-RPoint}\\
&l\bnd\icp k\pbnd
&&\al{=}(l\bnd\icp k)\pbnd
\label{writer:pbndbnd-assoc}\sctag{PBndBnd-Assoc}
}

But how to express the relative monad bind in terms of the bind of $\frm$? 
The functions $\shift$ and $\fuse$ studied in connection with mixmap can help. 
Firstly, we can define $\phi$ via $\hole\hdl$ by generalizing
\ref{writer:fun-hdl}:
\lawarray{
&\phi(g)
&&\al{=}(\rho\icp g)\hdl
\label{writer:rfun-hdl}\sctag{Mixmap-RBnd}
}
Then $\shift$ and $\fuse$ are expressible via $\phi$ by 
\ref{writer:shift-rfun} and \ref{writer:fuse-rfun}; 
hence these functions are definable in our
\qquot{two-story monad} framework. Now we achieve an
equivalent axiomatics in terms of monad $\frm$, $\shift$ and $\fuse$ 
if we add the following axioms, the first two of which are for defining 
$\hole\hdl$ and $\hole\pbnd$, to the union of monad axioms for $\frm$ 
and the previously seen axioms of $\shift$ and $\fuse$:
\lawarray{
&\shift\icp k\hdl
&&\al{=}(\shift\icp k)\bnd\icp\shift
\label{writer:shift-bndhom}\sctag{Shift-BndHom}\\
&k\pbnd
&&\al{=}k\bnd\icp\rho
\label{writer:pbnd-bnd}\sctag{PBnd-Bnd}\\
&\rho\bnd
&&\al{=}\fuse
\label{writer:bnd-rpoint}\sctag{Bnd-RPoint}
}
The previously established axiomatics for $\shift$ and $\fuse$ declares 
$\shift$ to be a homomorphism between the relative point and point, as well
as between mixmap and functor; \ref{writer:shift-bndhom} extends this 
pattern also to the third level. It enables to express $\hole\hdl$ via
$\hole\bnd$ as
\lawarray{
&k\hdl
&&\al{=}k\bnd\icp\shift
\label{writer:hdl-bnd}\sctag{RBnd-Bnd}
}
Indeed:
\prfarray{
&k\hdl\\
\al{=}&\since{\ref{writer:fuse-shift}, identity}\\
&\fuse\icp\shift\icp k\hdl\\
\al{=}&\since{\ref{writer:bnd-rpoint}, \ref{writer:shift-bndhom}}\\
&\rho\bnd\icp(\shift\icp k)\bnd\icp\shift\\
\al{=}&\since{monad}\\
&(\rho\bnd\icp\shift\icp k)\bnd\icp\shift\\
\al{=}&\since{\ref{writer:bnd-rpoint}}\\
&(\fuse\icp\shift\icp k)\bnd\icp\shift\\
\al{=}&\since{\ref{writer:fuse-shift}, identity}\\
&k\bnd\icp\shift
}
Note also that the laws imply $\rho$ being a monad morphism:
\ref{writer:point-rpoint} states preservation of unit, and
\ref{writer:pbnd-unitr} and \ref{writer:pbnd-bnd} together establish
preservation of bind.

Our main result of the current section 
is formalized by the following theorem, which again is correct for an 
arbitrary category:

\begin{theorem}\label{writer:level3}
Let $(\frj,\eta,\mu)$ be a monad, let $\frm$ be given as in
Theorem~\ref{writer:level2}, and let $\rho$, $\invok$, $\shift$, $\fuse$ and
$\phi$ be given with the same types as in Theorem~\ref{writer:level2}. 
Moreover, assume:
\lawarray{
&\hole\bnd
&&\al{\oftyp}(\obx\to\frm\ap\obx')\to\frm\ap\obx\to\frm\ap\obx'\\
&\hole\hdl
&&\al{\oftyp}(\frj\ap\obx\to\frm\ap\obx')\to\frm\ap\obx\to\frm\ap\obx'\\
&\hole\pbnd
&&\al{\oftyp}(\obx\to\frm\ap\obx')\to\frj\ap\obx\to\frm\ap\obx'
}
Then the set consisting of the laws
\ref{writer:hdl-unitl}, \ref{writer:hdl-id}, \ref{writer:hdl-assoc},
\ref{writer:pbnd-unitl}, \ref{writer:pbnd-unitr},
\ref{writer:pbndbnd-assoc}, \ref{writer:rfun-hdl}, \ref{writer:fun-hdl} 
(or: \ref{writer:fun-rfun}), \ref{writer:shift-rfun}, \ref{writer:fuse-rfun}, 
\ref{writer:point-rpoint} and \ref{writer:bnd-pbndhdl} is
equivalent to the axiomatics consisting of:
\begin{itemize}
\item Naturality of $\fuse$;
\item Three coherence laws of $\fuse$;
\item Three homomorphism laws of $\shift$
(\ref{writer:shift-rpointhom}, \ref{writer:shift-rfunhom} and
\ref{writer:shift-bndhom});
\item Four monad laws of $\frm$ (including the definition of the morphism
mapping of~$\frm$);
\item \ref{writer:bnd-rpoint} and \ref{writer:pbnd-bnd}.
\end{itemize}
\end{theorem}

The proofs are straightforward. In the light of the ability of expressing 
the operations used in this paper and the \ensuremath{\Conid{MonadWriter}} class methods in
terms of each other, we also have the following result:

\begin{theorem}
Let $\frm$ be any monad in the category $\Set$. Then: 
\begin{itemize}
\item The monad obtained by applying the writer monad transformer to~$\frm$,
with methods defined as in MTL, satisfies all laws mentioned in 
Theorem~\ref{writer:level3} and also \ref{writer:pbnd-bifun};
\item Applying the exception, reader, writer, and 
state monad transformers preserve the laws.
\end{itemize}
\end{theorem}

\subsection{Connections to Steele's Pseudomonads}\label{writer:steele}

Our two-story monads are close to the pseudomonad towers studied by 
Steele~Jr.~\cite{DBLP:conf/popl/Steele94}. That classic paper aims to find ways to
join different monadic effects by \qquot{composing} the carrier monads of each
particular effect. As monads are not always behaving nicely under composition, 
it introduces pseudomonads that generalize monads by allowing the target
type of the function under bind to differ from the source type of the
function produced by bind (the target types of the function under bind and
of that produced by bind still coincide). 
Unit and bind of a pseudomonad are called 
pseudounit and pseudobind. In this sense, our operation
$\hole\pbnd$ along with the unit $\eta$ of monad~$\frj$ 
are pseudobind and pseudounit of a pseudomonad. The notion of 
monad itself as treated in~\cite{DBLP:conf/popl/Steele94} 
is wider than standard and subsumes also relative monads. 

That paper modifies the standard monad axioms to
be applicable to pseudomonads. Here are Steele's axioms, written 
in the language of our paper:
\lawarray{
&k\pbnd\icp\eta
&&\al{=}k
\label{writer:steele-unitl}\sctag{Steele-UnitL}\\
&(h\icp\eta)\pbnd
&&\al{=}h
\label{writer:steele-unitr}\sctag{Steele-UnitR}\\
&l\pbnd\icp k^*
&&\al{=}(l\pbnd\icp k)\pbnd
\label{writer:steele-assoc}\sctag{Steele-Assoc}
}
Of these, \ref{writer:steele-unitl} coincides with our axiom
\ref{writer:pbnd-unitl}. But the next axiom, \ref{writer:steele-unitr}, is
not valid in general. It implies that every function of type 
$\frj\ap\obx\to\frm\ap\obx'$ is a result of pseudobind---and 
in particular, every
function of type $\frj\ap\obx\to\frj\ap\obx'$ is a result of bind, which is
clearly not true. The last axiom, \ref{writer:steele-assoc}, is a
theorem in our axiomatics, provable as follows:
\prfarray{
&l\pbnd\icp k^*\\
\al{=}&\since{\ref{writer:pbnd-bnd}}\\
&l\bnd\icp\rho\icp k^*\\
\al{=}&\since{\ref{writer:pbnd-unitr}}\\
&l\bnd\icp(\rho\icp k)\pbnd\\
\al{=}&\since{\ref{writer:pbndbnd-assoc}}\\
&(l\bnd\icp\rho\icp k)\pbnd\\
\al{=}&\since{\ref{writer:pbnd-bnd}}\\
&(l\pbnd\icp k)\pbnd
}
We were not able to prove \ref{writer:pbndbnd-assoc} from
\ref{writer:steele-assoc}, so it seems that \ref{writer:steele-assoc} is
strictly weaker than \ref{writer:pbndbnd-assoc}. (This does not mean a flaw
in~\cite{DBLP:conf/popl/Steele94}. The weaker axiom might
be perfect for the purposes of that paper which aims to 
involve cases where the monad obtained by composition does not satisfy
associativity. In our axiomatics, associativity of monad $\frm$ is 
forced by the laws of its constituent pseudomonad and relative monad.)

\subsection{Some Corollaries and Non-Corollaries}\label{writer:other}

In this subsection, we are working in the category $\Set$. 
According to the class definition given at the beginning of
Sect.~\ref{writer}, one can express
\[
\tell=\rho\icp(\constof\alone\both\id)\mbox{.}
\] 
Using the theory developed above, any expression of the form 
$\tellof w\sequ t$ can be
rewritten as $(\constof t)\pbnd\ap(\alone,w)$. (Rewrite $\tellof w$ and 
$\sequ$ by their meaning and apply \ref{writer:pbnd-bnd}.) 
On the other hand, \ref{writer:pbnd-bifun} allows to conclude that
\lawarray{
&(\constof t)\pbnd\ap(\alone,w)
&&\al{=}\frf\ap(\rsect{w}{\cdot},\id)\ap t
\label{writer:bifun-pbnd}\sctag{Bifun-PBnd}
}
Therefore, bifunctor applications of the form
$\frf\ap(\rsect{w}{\cdot},\id)\ap t$ 
are equivalent to $\tellof w\sequ t$. Using \ref{writer:bifun-pbnd} 
as the definition of such bifunctor applications, we can prove that they 
satisfy the functor laws. For identity:
\prfarray{
&\id\\
\al{=}&\since{identity, constant}\\
&\lam{t}{\constof t\ap\alone}\\
\al{=}&\since{\ref{writer:pbnd-unitl}}\\
&\lam{t}{(\constof t)\pbnd\ap(\eta\ap\alone)}\\
\al{=}&\since{writer}\\
&\lam{t}{(\constof t)\pbnd\ap(\alone,\one)}\\
\al{=}&\since{\ref{writer:bifun-pbnd}, extensionality}\\
&\frf\ap(\rsect{\one}{\cdot},\id)\\
\al{=}&\since{monoid unit}\\
&\frf\ap(\id,\id)
}
For composition,
\prfarray{
&\frf\ap(\rsect{w'}{\cdot},\id)\icp\frf\ap(\rsect{w}{\cdot},\id)\\
\al{=}&\since{composition}\\
&\lam{t}{\frf\ap(\rsect{w'}{\cdot},\id)\ap(\frf\ap(\rsect{w}{\cdot},\id)\ap t)}\\
\al{=}&\since{\ref{writer:bifun-pbnd} twice}\\
&\lam{t}{(\constof((\constof t)\pbnd\ap(\alone,w)))\pbnd\ap(\alone,w')}\\
\al{=}&\since{constant, product, $\alone$}\\
&\lam{t}{((\constof t)\pbnd\icp(\id\both\constof w))\pbnd\ap(\alone,w')}\\
\al{=}&\since{\ref{writer:steele-assoc}}\\
&\lam{t}{((\constof t)\pbnd\icp(\id\both\constof w)^*)\ap(\alone,w')}\\
\al{=}&\since{composition}\\
&\lam{t}{(\constof t)\pbnd\ap((\id\both\constof w)^*\ap(\alone,w'))}\\
\al{=}&\since{writer}\\
&\lam{t}{(\constof t)\pbnd\ap((\id\times\rsect{w'}{\cdot})\ap((\id\both\constof w)\ap\alone))}\\
\al{=}&\since{product, constant}\\
&\lam{t}{(\constof t)\pbnd\ap(\alone,w'\cdot w)}\\
\al{=}&\since{\ref{writer:bifun-pbnd}, extensionality}\\
&\frf\ap(\rsect{w'\cdot w}{\cdot},\id)\\
\al{=}&\since{monoid associativity}\\
&\frf\ap(\rsect{w'}{\cdot}\icp\rsect{w}{\cdot},\id)
}
The latter implies the equation 
\[
\tellof w\sequ\tellof w'=\tellof(w\cdot w')
\]
(for proof, rewrite $\tellof w'=\tellof w'\sequ\invokof\alone$ 
and later the same for $w\cdot w'$).

We leave the proofs of the following two corollaries as an exercise:
\lawarray{
&\rho\icp(\id\times\rsect{w}{\cdot})
&&\al{=}\frf\ap(\rsect{w}{\cdot},\id)\icp\rho
\label{writer:rpoint-binat}\sctag{RPoint-Binat}\\
&\frf\ap(\rsect{w}{\cdot},\id)\icp k\bnd
&&\al{=}k\bnd\icp\frf\ap(\rsect{w}{\cdot},\id)
\label{writer:bifun-bnd-comm}\sctag{Bifun-Bnd-Comm}
}

Finally, there are laws that cannot be deduced from the developed theory but
are valid in all monads constructible by applying the writer
monad transformer to any monad and preserved by the exception, 
reader, writer and state monad transformers. For instance:
\begin{itemize}
\item \ref{writer:rpoint-binat} stays true after replacing 
$\rsect{w}{\cdot}$ with arbitrary $f\oftyp\obw\to\obw$; 
\item For any monoid endomorphism~$h$, 
mapping of the computation log by~$h$ is a monad homomorphism:
\lawarray{
&\frf\ap(h,\id)\icp\invok
&&\al{=}\invok
\label{writer:bifun-unithom}\sctag{Bifun-UnitHom}\\
&\frf\ap(h,\id)\icp k\bnd
&&\al{=}(\frf\ap(h,\id)\icp k)\bnd\icp\frf\ap(h,\id)
\label{writer:bifun-bndhom}\sctag{Bifun-BndHom}
}
\end{itemize}

\section{Stateful Computations}\label{state}

In MTL, the class of stateful monads is introduced as follows:
\begin{hscode}\SaveRestoreHook
\column{B}{@{}>{\hspre}l<{\hspost}@{}}%
\column{5}{@{}>{\hspre}l<{\hspost}@{}}%
\column{14}{@{}>{\hspre}l<{\hspost}@{}}%
\column{17}{@{}>{\hspre}l<{\hspost}@{}}%
\column{19}{@{}>{\hspre}l<{\hspost}@{}}%
\column{E}{@{}>{\hspre}l<{\hspost}@{}}%
\>[B]{}\mathbf{class}\;\Conid{Monad}\;\Varid{m}\Rightarrow \Conid{MonadState}\;\Varid{s}\;\Varid{m}\mid \Varid{m}\to \Varid{s}\;\mathbf{where}{}\<[E]%
\\
\>[B]{}\hsindent{5}{}\<[5]%
\>[5]{}\Varid{get}{}\<[14]%
\>[14]{}\mathbin{::}\Varid{m}\;\Varid{s}{}\<[E]%
\\
\>[B]{}\hsindent{5}{}\<[5]%
\>[5]{}\Varid{get}{}\<[14]%
\>[14]{}\mathrel{=}{}\<[17]%
\>[17]{}\Varid{state}\;(\lambda \Varid{s}\to (\Varid{s},\Varid{s})){}\<[E]%
\\[\blanklineskip]%
\>[B]{}\hsindent{5}{}\<[5]%
\>[5]{}\Varid{put}{}\<[14]%
\>[14]{}\mathbin{::}\Varid{s}\to \Varid{m}\;(){}\<[E]%
\\
\>[B]{}\hsindent{5}{}\<[5]%
\>[5]{}\Varid{put}\;\Varid{s}{}\<[14]%
\>[14]{}\mathrel{=}{}\<[17]%
\>[17]{}\Varid{state}\;(\lambda \anonymous \to ((),\Varid{s})){}\<[E]%
\\[\blanklineskip]%
\>[B]{}\hsindent{5}{}\<[5]%
\>[5]{}\Varid{state}{}\<[14]%
\>[14]{}\mathbin{::}(\Varid{s}\to (\Varid{a},\Varid{s}))\to \Varid{m}\;\Varid{a}{}\<[E]%
\\
\>[B]{}\hsindent{5}{}\<[5]%
\>[5]{}\Varid{state}\;\Varid{f}{}\<[14]%
\>[14]{}\mathrel{=}{}\<[17]%
\>[17]{}\mathbf{do}{}\<[E]%
\\
\>[17]{}\hsindent{2}{}\<[19]%
\>[19]{}\Varid{s}\leftarrow \Varid{get}{}\<[E]%
\\
\>[17]{}\hsindent{2}{}\<[19]%
\>[19]{}\mathbf{let}\mathord{\sim}(\Varid{a},\Varid{s'})\mathrel{=}\Varid{f}\;\Varid{s}{}\<[E]%
\\
\>[17]{}\hsindent{2}{}\<[19]%
\>[19]{}\Varid{put}\;\Varid{s'}{}\<[E]%
\\
\>[17]{}\hsindent{2}{}\<[19]%
\>[19]{}\Varid{return}\;\Varid{a}{}\<[E]%
\ColumnHook
\end{hscode}\resethooks
The method \ensuremath{\Varid{get}}, by intention, returns the current state
without modifying it, and the method \ensuremath{\Varid{put}} replaces the current state 
with the one given as argument while returning the trivial value \ensuremath{()}.
The method \ensuremath{\Varid{state}}, analogously to \ensuremath{\Varid{reader}} and \ensuremath{\Varid{writer}} seen
previously, embeds a stateful operation represented in the form of a pure
function in the monad. 

If we denote by $\frj$ the state monad
$\frj\ap\obx=\obs\to\obx\times\obs$, where $\obs$ is an arbitrary fixed
type, and by $\frm$ any monad of the \ensuremath{\Conid{MonadState}} class, then the method 
\ensuremath{\Varid{state}} has type $\frj\ap\obx\to\frm\ap\obx$ and is in principle the
point $\rho$ of the functor $\frm$ relatively pointed on $\frj$. 
Thus all methods of \ensuremath{\Conid{MonadState}} 
classify as first-level in the hierarchy defined
in Sect.~\ref{except} as $\gt$ and $\pt$ are special cases of $\rho$. 
Like in Sect.~\ref{writer}, we keep the second type fixed,
i.e., functors have only one argument.

Compared to the other classes of monads, there exists relatively much 
previous research considering some equational axiomatization of \ensuremath{\Conid{MonadState}}. 
Gibbons and Hinze~\cite{DBLP:conf/icfp/GibbonsH11} assume the following set
of axioms:
\lawarray{
&\ptof s\sequ\ptof s'
&&\al{=}\ptof s'
\label{state:put-put}\sctag{Put-Put}\\
&\ptof s\sequ\gt
&&\al{=}\ptof s\sequ\invokof s
\label{state:put-get}\sctag{Put-Get}\\
&\gt\bind\pt
&&\al{=}\invokof()
\label{state:get-put}\sctag{Get-Put}\\
&\gt\bind\lam{s}{\gt\bind k\ap s}
&&\al{=}\gt\bind\lam{s}{k\ap s\ap s}
\label{state:get-get}\sctag{Get-Get}
}
These axioms, which can now be called classic, 
are valid for all monads obtained by application of the
state monad transformer to any monad and subsequent applications of error,
reader, writer and state monad transformers. 
Almost all other work that we are aware of relies on the same axiomatics with 
irrelevant modifications. Exceptions are Harrison and 
Hook~\cite{DBLP:conf/csfw/HarrisonH05} and 
Harrison~\cite{DBLP:conf/aplas/Harrison06} that formalize their axiomatics
in terms of $\gt$ and $\updt$ rather than $\gt$ and $\pt$, where 
$\updt$ corresponds to the MTL function \ensuremath{\Varid{modify}} defined by
\begin{hscode}\SaveRestoreHook
\column{B}{@{}>{\hspre}l<{\hspost}@{}}%
\column{11}{@{}>{\hspre}l<{\hspost}@{}}%
\column{E}{@{}>{\hspre}l<{\hspost}@{}}%
\>[B]{}\Varid{modify}{}\<[11]%
\>[11]{}\mathbin{::}\Conid{MonadState}\;\Varid{s}\;\Varid{m}\Rightarrow (\Varid{s}\to \Varid{s})\to \Varid{m}\;(){}\<[E]%
\\
\>[B]{}\Varid{modify}\;\Varid{f}{}\<[11]%
\>[11]{}\mathrel{=}\Varid{state}\;(\lambda \Varid{s}\to ((),\Varid{f}\;\Varid{s})){}\<[E]%
\ColumnHook
\end{hscode}\resethooks
(The early papers on transformers 
\cite{DBLP:conf/popl/LiangHJ95,DBLP:conf/afp/Jones95,HuttonM1996} had a
singleton method \ensuremath{\Varid{update}} in class 
\ensuremath{\Conid{MonadState}}	; this method behaved
similarly to \ensuremath{\Varid{modify}} but returned the state.) 
The axiomatics used by 
\cite{DBLP:conf/csfw/HarrisonH05,DBLP:conf/aplas/Harrison06} seem to be 
weaker than
that of \cite{DBLP:conf/icfp/GibbonsH11} and our experience suggests that 
any set of laws for \ensuremath{\Varid{modify}} 
equivalent to the Gibbons-Hinze axiomatics of \ensuremath{\Varid{get}} and
\ensuremath{\Varid{put}} is probably less elegant.

Sometimes also the following unit law is included:
\lawarray{
&\gt\sequ t
&&\al{=}t
\label{state:get-unitr}\sctag{Get-UnitR}
}
The law \ref{state:get-unitr} follows from the axioms above. More
surprisingly, the last among the above axioms, 
\ref{state:get-get}, turns out to be implied by \ref{state:put-put}, 
\ref{state:put-get} and \ref{state:get-put}. Given the translations between
the methods of class \ensuremath{\Conid{MonadState}} in the class declaration, 
the set of axioms \ref{state:put-put}, \ref{state:put-get} and 
\ref{state:get-put} is equivalent to the statement that 
$\rho\oftyp\frj\ap\obx\to\frm\ap\obx$ is a monad morphism. 
All the facts listed in this paragraph were
noted by Li-yao Xia in the post~\cite{Libraries} (we have checked that these
claims are correct indeed). As we have found no research 
papers mentioning them, despite the numerous authors using the 
state monad axioms, it seems that these facts are not widely known 
to the Haskell research community.

\section{Conclusion}\label{conc}

We have investigated equational axiomatizations for the \ensuremath{\Conid{MonadError}},
\ensuremath{\Conid{MonadReader}}, \ensuremath{\Conid{MonadWriter}}, and \ensuremath{\Conid{MonadState}} type classes
defined in Haskell MTL, along with reviewing previous related work. 
For each class, we have (or the previous 
work referred to has) proposed at least two alternative axiomatics that are
proven to be equivalent. In the case of \ensuremath{\Conid{MonadError}} and
\ensuremath{\Conid{MonadReader}}, the proposed axiomatics assume a wider setting that 
goes beyond the limits imposed by the Haskell MTL. We think that 
this should not be considered as a shortcoming, as the opportunity to extend
the world can facilitate proving theorems or 
perhaps even prove more theorems about the usual Haskell world. For
instance, it is not easy to find a proof of the analogous 
to \ref{state:get-unitr} law
\lawarray{
&\ask\sequ t
&&\al{=}t
\label{reader:ask-unitr}\sctag{Ask-UnitR}
}
using only the Haskell world laws of \ensuremath{\Conid{MonadReader}} from
Subsect.~\ref{reader:reader} and the definition of $\ask$. 
When switching to the axiomatics that involves $\apply$ and $\abstr$, 
proving \ref{reader:ask-unitr} becomes straightforward.

One more contribution of this work is a classification of all methods
of the four MTL classes into three levels (points, mixmaps and handles)
which have similar categorical interpretation for all classes and based of
which axiomatics of different classes can be compared. 
For instance, a lot of laws are common to the axiomatics of 
\ensuremath{\Conid{MonadError}} studied in our previous work~\cite{DBLP:conf/ictac/Nestra19} 
and the axiomatics of \ensuremath{\Conid{MonadWriter}} considered in this paper; some of the
common laws may look different. 
One can observe that unit and join of the exception monad 
$\Except\ap\obx=\obe+\obx$ are $\inr$ and $\inl\either\id$,
respectively; so one could define $\shift=\phi(\inr)$ and
$\fuse=\phi(\inl\either\id)$ for exceptions like we did for writer effects. 
In~\cite{DBLP:conf/ictac/Nestra19}, $\phi(\inl\either\id)$ was denoted
$\fusel$, and $\phi(\inr)$ is equivalent to 
$\rho\cch\icp\frf\ap(\inr\icp\inl,\inr)$ 
in the axiomatics of~\cite{DBLP:conf/ictac/Nestra19}. Hence our two-story
monad laws \ref{writer:fuse-funpoint} and \ref{writer:hdl-bnd} 
occurred in~\cite{DBLP:conf/ictac/Nestra19} in a form that did not help to 
recognize them as phenomenons of such an abstract level. 
We hope that this uniform view of different effects may conduce to deeper 
understanding of the MTL class methods. 

The paper might also inspire future development of MTL. 
Using bifunctors rather than monads or replacing MTL methods
with functions more convenient in the theory is not meant 
to suggest making the corresponding changes in Haskell. 
For instance, the method \ensuremath{\Varid{listen}} is likely to be more useful in practice 
than our $\shift$ would be, and we do not know a way to redefine all MTL 
classes simultaneously to apply to binary type constructors instead of 
unary ones without severe penalties in usability.
(Our earlier work~\cite{DBLP:conf/ictac/Nestra18} makes steps towards the 
latter, with the aim of increasing the expressive power of 
exception handling in the applicative style, but 
the only type that are added to monads as an extra parameter 
is the type of exceptions.)
However, adding some new methods to the existing
classes would be reasonable. In particular, the class \ensuremath{\Conid{MonadWriter}}
would benefit from including the general mixmap as class method, as it would
provide the shortest way to define various functions of the mixmap level
one might desire besides the existing \ensuremath{\Varid{listen}} and \ensuremath{\Varid{pass}} methods. 
The same addition could be useful in class \ensuremath{\Conid{MonadError}} 
(with fixed error type).

%
\bibliographystyle{ACM-Reference-Format}
\bibliography{nestra}

\appendix

\end{document}